\documentclass[journal,twoside,web]{IEEEtran}
\usepackage{cite}
\usepackage{amsmath,amssymb,amsfonts}
\usepackage{algorithmic}
\usepackage{graphicx}
\usepackage{algorithm,algorithmic}
\usepackage{hyperref}
\hypersetup{hidelinks=true}
\usepackage{textcomp}

\usepackage{graphics}
\usepackage{graphicx}
\graphicspath{{./fig/}}
\usepackage{color}
\usepackage{subfig}
\usepackage{amssymb}
\usepackage{amsmath}
\allowdisplaybreaks
\usepackage{amsfonts}
\usepackage{mathrsfs}
\usepackage{latexsym, bm}
\usepackage{enumerate}
\usepackage{xcolor}
\usepackage{pdfpages}
\usepackage{amsthm}
\newtheorem{theorem}{Theorem}
\newtheorem{proposition}{Proposition}

\newtheorem{lemma}{Lemma}
\newtheorem{corollary}{Corollary}

\newtheorem{remark}{Remark}
\newtheorem{example}{Example}
\newtheorem{assumption}{Assumption}
\newtheorem*{problem}{Problem}

\newcommand{\amax}{a_{\max}}
\newcommand{\amin}{a_{\min}}

\newcommand{\interior}{\operatorname{int}}
\newcommand{\norm}[1]{\| #1 \|}
\newcommand{\Vf}{\mathcal{V}_{\textup{f}}}
\newcommand{\Vp}{\mathcal{V}_{\textup{p}}}
\newcommand{\V}{\mathcal{V}}

\DeclareMathOperator{\diag}{diag}

\newcommand{\until}[1]{\{1,\dots, #1\}}

\newcommand{\setdef}[2]{\{#1 \; | \; #2\}}

\newcommand{\map}[3]{#1: #2 \rightarrow #3}

\newcommand{\intersection}{\ensuremath{\operatorname{\cap}}}

\newcommand{\adj}{\textup{adj}}

\newcommand\oprocendsymbol{\hbox{$\square$}}
\newcommand\oprocend{\relax\ifmmode\else\unskip\hfill\fi\oprocendsymbol}

\renewcommand{\theenumi}{(\roman{enumi})}

\DeclareSymbolFont{bbold}{U}{bbold}{m}{n}
\DeclareSymbolFontAlphabet{\mathbbold}{bbold}

\newcommand{\db}[1]{\bm{[}#1\bm{]}}

\usepackage[prependcaption,colorinlistoftodos]{todonotes}

\begin{document}
\title{Dynamical models for distributed social power perception in Friedkin-Johnsen influence networks}
\author{Ye Tian, Angela Fontan, Yu Kawano, Wei Zhang, Kenji Kashima, and Karl H. Johansson
\thanks{This work was partially supported by 
  the JSPS KAKENHI (Grant Nos. JP21H04875 and JP21K14185), the JST ASPIRE Program (Grant No. JPMJAP2402), the Wallenberg AI, Autonomous Systems and Software Program (WASP) and the Wallenberg Scholar Grant funded by the Knut and Alice Wallenberg Foundation, and the Swedish Research Council Distinguished Professor Grant.}
\thanks{Y. Tian and K. Kashima are with the Graduate School of Informatics, Kyoto University, Kyoto 606-8501, Japan (e-mail: tinybeta7.1@gmail.com; kk@i.kyoto-u.ac.jp).}
\thanks{A. Fontan and K. H. Johansson are with the Division of Decision and Control Systems, School of Electrical Engineering and Computer Science, KTH Royal Institute of Technology, and with Digital Futures, Stockholm 100 44, Sweden (e-mail: \{angfon; kallej\}@kth.se). }
\thanks{Y. Kawano is with the Graduate School of Advanced Science and Engineering, Hiroshima University, Higashi-Hiroshima 739-8527, Japan (e-mail: ykawano@hiroshima-u.ac.jp).}
\thanks{W. Zhang is with the Department of Mechanical and Energy Engineering, Southern University of Science and Technology, Shenzhen 518055, China (e-mail: zhangw3@sustech.edu.cn).}}

\maketitle

\begin{abstract}
Social power quantifies the ability of individuals to influence others and plays a central role in social influence networks. Yet, computing social power typically requires global knowledge and significant computational or storage capability, especially in large-scale networks with stubborn individuals. In this paper, we propose a distributed perception mechanism based on the Friedkin-Johnsen opinion dynamics that enables individuals to estimate their true social power through local interactions. The mechanism starts from independent initial perceptions and relies only on local information: each individual only needs to know its neighbors' stubbornness and the influence weights they accord. We provide rigorous dynamical system analysis that characterizes equilibria, invariant sets, and convergence. Conditions are established for convergence to the true social power in both the static setting with fixed influence weights and the reflected-appraisal setting where influence weights coevolve with perceptions. The proposed mechanism remains reliable under extreme initial perceptions, disconnected influence networks, reflected-appraisal coupling, and variations in timescales. Numerical examples illustrate our results. 
\end{abstract}



\section{Introduction}\label{s1}
Social power refers to the ability of an individual to control the thoughts, feelings, or behaviors of others~\cite{JRPF:56,DC:59}, and has been identified with influence, competence, knowledge, and authority~\cite{RB:50}. In influence system theory and opinion dynamics, social power plays a crucial role in shaping opinion outcomes, restructuring interpersonal influence, and reflecting the level of collective intelligence~\cite{YT-LW:23}.

Models of opinion dynamics describe how individuals' opinions evolve under interpersonal influence,
including the DeGroot model~\cite{MHDG:74}, the Friedkin-Johnsen (FJ) model~\cite{NEF-ECJ:99}, and various extensions~\cite{AVP-RT:18,YT-LW:23}.   
A mathematical formalization of social power was proposed in~\cite{NEF:11} based on the FJ model, where an individual's social power is defined as the normalized total contribution of their initial opinion to the final opinions of all individuals. This definition was later extended to the general framework of weighted-average opinion dynamics~\cite{YT-LW-FB:23}. In this framework, social power is a probability vector whose components  
capture the contributions of individuals' initial opinions to the average of their final opinions, providing a measure of influence centrality. 

Social power is critical in determining whether an influence system is wise~\cite{BG-MOJ:10, FB-FF-BF:20}, and whether local interactions improve or undermine collective wisdom~\cite{YT-LW-FB:23}. It also underlies a variety of higher-level dynamics. For instance, in social power games~\cite{10506633,CB-LW-MF-CA:23}, an agent's social power serves as their utility. The evolution of social power has been studied based on the psychological theory of \emph{reflected appraisal}~\cite{NEF-PJ-FB:16}, most notably through the DeGroot-Friedkin (DF) model~\cite{PJ-AM-NEF-FB:13d} and its extensions~\cite{MY-JL-BDOA-CY-TB:17b,YT-PJ-AM-LW-NEF-FB:22}. In these models, individuals update their self-weights over successive rounds of opinion formation according to their manifested social power. 

An implicit assumption in much of the existing literature is that individuals have accurate knowledge of their social power~\cite{JK-JO-AS-PBG:25}. However, social power is a latent quantity whose mathematical definition typically requires global knowledge of \emph{all} model parameters, even to compute the power of a single individual~\cite{YT-LW-FB:23}. This challenge is further exacerbated in the presence of stubborn individuals~\cite{YT-PJ-AM-LW-NEF-FB:22}, where additional computational or storage resources are required. These requirements make it particularly challenging to compute social power in large-scale networks. How individuals come to know their social power remains an open and important question.

Empirical evidence suggests that humans possess neurocognitive capacities to perceive observable properties of their social environments~\cite{MG-WBB-JD-SLF-FK-HO-DP-DLS-TD:21}. In~\cite{AZ-ARC-ES:59}, social power is termed \emph{perceived relative power}, defined as the ability individuals perceive themselves to have to influence others. This notion of perceived social power is conceptually distinct from \emph{true social power}: 
the latter represents the actual influence an individual exerts on others, whereas perceived social power reflects an internal estimate of this influence. While these studies suggest that individuals can autonomously estimate their social power, rigorous mathematical models explaining how such perception can arise from local interactions remain limited~\cite{JK-JO-AS-PBG:25}, especially in the presence of stubborn individuals.

A brief discussion of distributed social power perception in the DeGroot framework appears in~\cite{PJ-AM-NEF-FB:13d}, where successful perception relies on consensus-type network connectivity and normalized initialization. Related ideas are employed in~\cite{PJ-NEF-FB:14e,XC-JL-MAB-ZX-TB:16} to study single-timescale variants of the DF model and distributed evaluation of self-appraisals under the reflected-appraisal theory. Although not explicitly concerned with social power perception, these models can be interpreted as perception dynamics with reflected appraisals: individuals update their self-appraisals using local variables instead of their true social power.
However, their scope and theoretical guarantees are limited. In particular, it is unclear in~\cite{XC-JL-MAB-ZX-TB:16} whether the local variables converge to the true social power equilibrium, while~\cite{PJ-NEF-FB:14e} guarantees such convergence only for irreducible networks. Moreover, these models assume the initial perception vector lies in the probability simplex, which requires globally coordinated initialization. Importantly, 
the underlying opinion dynamics in these works follows the DeGroot model, where social power depends only on the influence network. 

In contrast, social power in the FJ model depends on both 
the influence network and individuals' stubbornness, which makes its computation more involved. Empirical studies suggest that the FJ model more accurately captures human opinion formation~\cite{NEF-ECJ:99}. Yet, how stubborn individuals perceive their social power in the FJ model remains unexplored.

\subsubsection*{Contribution}
In this paper, we propose a distributed perception dynamics that allows individuals to estimate their true social power in the FJ model. The proposed dynamics evolves perceived rather than true social power, thus does not modify the underlying opinion dynamics or the evolution of true social power. Instead, it operates on a separate layer and can be viewed as a \emph{distributed observer} for the latent true social power. During the perception process, individuals independently initialize their perceptions and update them via local interactions, where  
they need only the stubbornness of and the influence weights accorded by their neighbors. 

We first study a basic setting for 
perceiving the static true social power of the FJ model under fixed influence networks. We prove that the perceived social power converges exponentially to the true social power from arbitrary initial perceptions (Proposition~\ref{P1}). Notably, unlike the DeGroot perception mechanism, 
the proposed perception dynamics preserves neither the probability simplex nor the unit hypercube, and its convergence requires neither normalized initial perceptions nor connectivity assumptions on the influence network. These properties enable fully independent initializations and endow the perception dynamics with strong robustness to extreme initial perceptions and disconnected influence networks.

We then incorporate the perception dynamics with the reflected-appraisal mechanism, yielding a \emph{perceived-reflected-appraisal} dynamics in which individuals update their self-appraisals using perceived social power. 
In the probability simplex, we show that this dynamics shares the same equilibrium set as the evolution of true social power and establish uniqueness of the equilibrium (Proposition~\ref{P2} and Theorem~\ref{T2}). Together with the equivalence of equilibrium sets, this also resolves the uniqueness conjecture for the true social power dynamics. In addition, we derive a necessary condition for the existence of a dominant individual who holds a majority of social power at equilibrium (Proposition~\ref{P3}).
We further characterize two classes of positively invariant hyperrectangles and derive contraction conditions that guarantee convergence to the equilibrium of true social power (Theorem~\ref{T6}). 

Finally, we examine the applicability of our contraction conditions in structured FJ networks, including star topologies and homogeneous stubbornness. We show that the contraction conditions hold in representative FJ networks that exhibit either democratic or maximally nonuniform true social power structures (Corollaries~\ref{C02} and~\ref{C01}). We further exploit these structures to derive sharper convergence results, where weaker conditions guarantee convergence from broader sets of initial perceptions (Theorems~\ref{T3},~\ref{T4}, and~\ref{T5}). 

We believe that our investigation represents the first rigorous study of distributed social power perception in networks of stubborn individuals. Theoretical results show that the proposed dynamics enables effective and efficient perception of social power, and remains reliable under extreme initial perceptions, disconnected networks, reflected-appraisal coupling, and timescale variations. A key ingredient in our perception weights is the \emph{relative stubbornness} term, which fundamentally distinguishes the updates from the DeGroot case. As a result, the perception dynamics can not be obtained from or reduced to existing perception models in the DeGroot framework~\cite{PJ-NEF-FB:14e, XC-JL-MAB-ZX-TB:16}. In particular, the proposed FJ perception mechanism (i)  captures a richer and more realistic opinion formation process; (ii) allows independently generated, possibly extreme initial perceptions and applies to disconnected influence networks; and (iii) remains effective under reflected-appraisal dynamical influence networks, irrespective of network connectivity.

\subsubsection*{Organization}
Section \ref{s2} motivates and formulates the problem. Section \ref{s3} studies distributed perception of static true social power in the FJ model. Its extension with reflected appraisals is analyzed in Section \ref{s4}. Section \ref{s5} investigates the perception dynamics in structured networks. Section \ref{s6} concludes and discusses future directions. 
\subsubsection*{Notation}
Let $\mathbf{1}_{n}$ and $I_{n}$ denote the $n\times 1$ all-ones vector and the $n\times n$ identity matrix, respectively. $\mathbf{e}_{i}$ denotes the $i$-th standard basis vector of the proper dimension. Given $x\in\mathbb{R}^{n}$, $\db{x}=\diag(x)$ denotes a diagonal matrix with diagonal elements $x_{1}, \dots, x_{n}$. The $n$-simplex is denoted by $\Delta_{n}=\setdef{x\in\mathbb{R}^{n}}{x\geq 0, \mathbf{1}^{\top}_{n}x=1}$ and $\interior{\Delta_{n}}$ denotes its interior. For $y,z\in\mathbb{R}^{n}$ with $y\leq z$, denote $\varGamma_{n}(y,z)=\setdef{x\in\mathbb{R}^{n}}{y\leq x\leq z}$ the hyperrectangle and $\interior{\varGamma_{n}(y,z)}$ its interior. A nonnegative matrix is row-stochastic if its row sums are $1$; it is doubly-stochastic if both its row and column sums are $1$. The weighted digraph $\mathcal{G}(W)$ associated with a nonnegative matrix $W$ has node set $\V=\until n$ and a directed edge $(i,j)$ from nodes $i$ to $j$ if and only if $W_{ij}>0$, where $i$ (resp. $j$) is called the in-neighbor (resp. out-neighbor) of $j$ (resp. $i$). $\mathcal{N}_{i}^{+}$ and $\mathcal{N}_{i}^{-}$ denote the sets of in-neighbors and out-neighbors of $i$, respectively. $\mathcal{G}(W)$ is a star topology if all its directed edges are either from or to a (center) node. A directed path $q^{i_{0}i_{m}}$ of length $m$ from nodes $i_{0}$ to $i_{m}$ in $\mathcal{G}(W)$ consists of a finite and ordered sequence of distinct nodes $i_{0}, i_{1}, \dots, i_{m}$ satisfying $W_{i_{l}i_{l+1}}>0$ for all $0\leq l\leq m-1$. With a slight abuse of notation, we denote $q^{i_{0}i_{m}}=(i_{0}, i_{1}, \dots, i_{m})$ and $i_{l}\in q^{i_{0}i_{m}}$ for all $0\leq l\leq m$. Particularly, if $i_{0}=i_{m}$, the directed path is called a cycle of $i_{0}$ and is simply denoted by $q^{i_{0}}$. 

\section{Problem formulation}\label{s2}
This section presents the problem formulation.We begin by reviewing the notion of social power in the FJ model and its evolution under reflected-appraisal theory, which naturally motivates the distributed perception problem.

\subsection{Friedkin-Johnsen opinion dynamics}
Consider $n\geq 2$ individuals forming their opinions in an influence network $\mathcal{G}(W)$, where $W$ is the row-stochastic influence matrix. 
Let $y_{i}(k)\in\mathbb{R}$ denote the opinion of individual $i$ at time $k$. In the FJ model, individuals update their opinions according to
\begin{equation}\label{e2}
y(k+1)=A W y(k)+(I_{n}-A)y(0),
\end{equation}
where $y=(y_{1}, \dots, y_{n})^{\top}$ and $A=\diag(a_{1},\dots, a_{n})$. In~\eqref{e2}, $a_{i}\in[0,1]$ is individual $i$'s susceptibility to social influence, while $1-a_{i}$ represents its stubbornness to initial opinion. We say $i$ is \emph{stubborn} if $a_{i}<1$, and denote by $\Vf$ and $\Vp$ the sets of \emph{fully stubborn} ($a_{i}=0$) and \emph{partially stubborn} ($0<a_{i}<1$) individuals, respectively. 


\subsection{True social power and its evolution}
The FJ model~\eqref{e2} converges to the limit
\begin{align*}
  \lim_{k\to\infty}y(k)=(I_{n}-AW)^{-1}(I_{n}-A)y(0)
\end{align*}
if and only if $\rho(AW)<1$~\cite[Lemma 2]{YT-LW:18}. Let $V=(I_{n}-AW)^{-1}(I_{n}-A)$, where $V_{ji}$ captures the influence of $i$'s initial opinion on $j$'s final opinion. Accordingly, individuals' \emph{true social power}, denoted by $x=(x_{1}, \dots, x_{n})^{\top}$, is defined as the relative control of their initial opinions on all individuals' final opinions~\cite{NEF:11} and is given by
\begin{equation}\label{e4}
x=\frac{V^{\top}\mathbf{1}_{n}}{n}=(I_{n}-A)(I_{n}-W^{\top}A)^{-1}\frac{\mathbf{1}_{n}}{n}.
\end{equation}

The true social power~\eqref{e4} is the \emph{static} power with fixed influence matrix $W$. A well-established mechanism in influence system theory is reflected appraisal~\cite{NEF-PJ-FB:16}, which posits that individuals update their self-appraisals based on their social power. These evolving self-appraisals then reshape interpersonal influence and drive the evolution of social power~\cite{PJ-AM-NEF-FB:13d,YT-PJ-AM-LW-NEF-FB:22}. 
Let $\gamma\in\mathbb{R}^{n}$ denote the diagonal of $W$, representing individuals' self-appraisals. Then $W$ can be decomposed as $W=[\gamma]+(I_{n}-[\gamma])C$, where $C$ is the row-stochastic and zero-diagonal  relative influence matrix. Consider the FJ model evolving over a sequence of issues $s=0,1,\dots$:
\begin{equation}\label{e2.1}
y(s,k+1)=A W(\gamma(s)) y(s,k)+(I_{n}-A)y(s,0),
\end{equation}
where $y(s,k)$ is the opinion vector on issue $s$ at time $k$. Under the reflected-appraisal mechanism, individuals take their manifested social power over the prior issue as their self-appraisals on the next issue, i.e., $\gamma(s+1)=x(s)$. As a result, individuals' true social power evolves over the issue sequence as follows:
\begin{equation}\label{e5}
x(s+1)=(I_{n}-A)(I_{n}-W^{\top}(x(s))A)^{-1}\frac{\mathbf{1}_{n}}{n}.
\end{equation}
A single-timescale dynamics is also proposed in~\cite{YT-PJ-AM-LW-NEF-FB:22}, where reflected appraisal is incorporated after each opinion update such that social power evolves on a single issue:
\begin{equation}\label{e5.1}
  \begin{cases}
  V(k+1)=AW(x(k))V(k)+I_{n}-A,\\
  x(k+1)=V^{\top}(k+1)\frac{\mathbf{1}_{n}}{n},
  \end{cases}
\end{equation}
with $V(0)=I_{n}$. The equilibria and convergence properties of systems~\eqref{e5} and~\eqref{e5.1} are summarized in Lemma~\ref{L1} in Appendix~\ref{app:aux}.
Note that in~\eqref{e4},~\eqref{e5}, and~\eqref{e5.1}, $x_{i}\equiv 0$ if $a_{i}=1$. 
For simplicity, we make the following assumption. 

\begin{assumption}\label{A1}
The susceptibility $a_{i}<1$ for all $i\in\V$, and $a_{i}>0$ for at least one $i\in\V$.
\end{assumption}

Under Assumption~\ref{A1}, we have $\rho(AW)<1$.

\subsection{Perceived social power and distributed perception}

\begin{figure}[t]
  \centerline{\includegraphics[width=0.9\hsize]{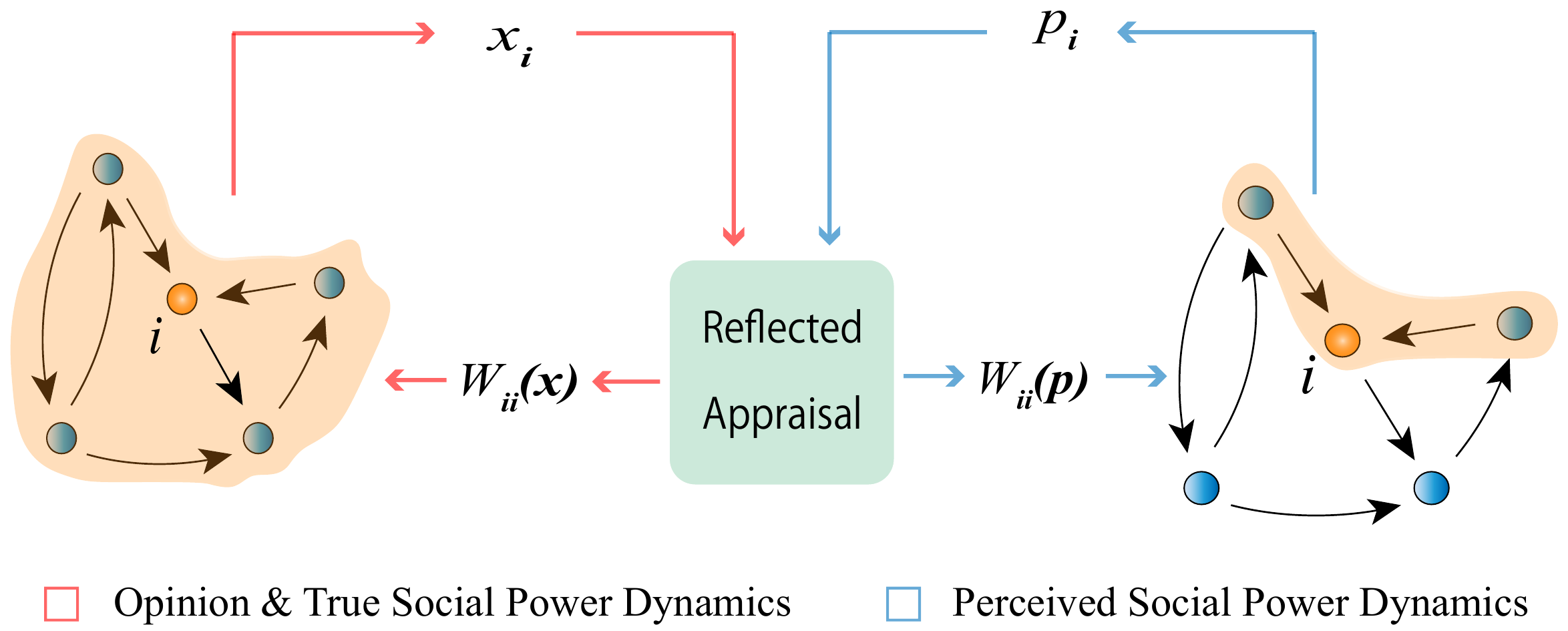}}
  \caption{Dynamics of true and perceived social power under reflected appraisal. The shaded region marks the information required to compute $x_i$ (true social power) or to update $p_{i}$ (perceived social power) in each layer.}
  \vspace{-1em}
  \label{fig0}
  \end{figure} 

The true social power~\eqref{e4} and its evolution in~\eqref{e5} and~\eqref{e5.1} are formulated in strict accordance with the definition $x=V^{\top}\mathbf{1}_{n}/n$. However, computing the true social power of even a single individual requires comprehensive global information, including the entire influence matrix $W$ and all susceptibilities $a_i$. Moreover, in~\eqref{e4} and~\eqref{e5}, evaluating social power entails computing the matrix inverse $(I_{n}-W^{\top}A)^{-1}$.
In~\eqref{e5.1}, additional storage requirements arise, as each individual must maintain a column of $V$, whose update depends on global information of $C$, $A$, and $x$. Overall, these requirements are rarely feasible in practice, particularly in large-scale networks. 

Motivated by these challenges and the notion of perceived social power~\cite{AZ-ARC-ES:59}, we develop a mathematical model for distributed social power perception in the FJ model. As illustrated in Fig.~\ref{fig0}, we first construct a perception layer separate from the opinion-dynamics layer, in which perceived social power evolves to estimate the static true social power using only local information. We then couple this perception layer with the reflected-appraisal mechanism and study whether it can track the equilibrium of the true social power dynamics. The problem of interest is formulated as follows.

\begin{problem}(Distributed social power perception)
Let $x$ denote the true social power of the FJ model defined by~\eqref{e4}, and let $x^{*}$ be the equilibrium of the true social power dynamics~\eqref{e5} and~\eqref{e5.1}. Let $p=(p_1,\dots,p_n)^\top$ denote the vector of perceived social power. We aim to design a distributed perception dynamics with local update rules of the form 
\begin{align}
  p_{i}(k+1)=\mathcal{P}_{i}(\{a_{j},W_{ji}(k),p_{j}(k)\}_ {j\in\mathcal{N}_{i}^{+}}),
\end{align}
such that 
\begin{enumerate}
  \item if $W(k)=W$ for all $k\geq 0$, the perceived social power converges to the static true social power, i.e., $\lim_{k\to\infty}p(k)=x$; \label{pro1}
  \item if individuals update their self-appraisals using their perceived social power, i.e., $W(k)=W(p(k))$, the perceived social power converges to the equilibrium of the true social power dynamics, i.e., $\lim_{k\to\infty}p(k)=x^{*}$. \label{pro2}
\end{enumerate} 
\end{problem}


\section{Distributed perception of static social power}\label{s3}
We first study distributed perception of the static true social power~\eqref{e4} and address Problem~\ref{pro1}. 
Assume that each individual $i$ knows the susceptibilities of their in-neighbors and the influence weights accorded by them, namely $a_{j}$ and $W_{ji}$ for all $j\in\mathcal{N}_{i}^{+}$. Starting from an arbitrary initial perception $p_{i}(0)\in \mathbb{R}$, individuals update their perceived social power according to the following local rule:
\begin{align}\label{e6}
p_{i}(k+1)=(1-a_{i})\sum_{j=1}^{n}\frac{a_{j}}{1-a_{j}}W_{ji}p_{j}(k)+\frac{1-a_{i}}{n}.
\end{align}

The perception dynamics~\eqref{e6} can be interpreted as a \emph{dynamical flow system}: $i$ updates their perceptions by aggregating that of all their in-neighbors $j\in\mathcal{N}_{i}^{+}$ with \emph{perception weights} $\frac{1-a_{i}}{1-a_{j}}a_{j}W_{ji}$, while the constant term $(1-a_{i})/n$ represents the baseline power $i$ retains by adhering to their initial opinion. 
The perception weight reflects how much $i$ believes their influence on $j$ contributes to their power. Specifically, $a_{j}W_{ji}$ is the effective influence weight $j$ accords to $i$ in the opinion dynamics~\eqref{e2}, while $\frac{1-a_{i}}{1-a_{j}}$ measures the \emph{relative stubbornness} of $i$ compared to $j$. If $1-a_{i}>1-a_{j}$, the perception weight $i$ assigns to $j$ exceeds the effective influence weight $j$ accords to $i$ in opinion formation. In this sense, a more stubborn individual tends to perceive themselves as having greater social power relative to less stubborn neighbors.

This interpretation highlights that the perception dynamics~\eqref{e6} is not a simple reversal of the opinion dynamics; rather, it reweights influence through the asymmetry introduced by relative stubbornness. This perception structure 
is fundamentally distinct from the DeGroot perception dynamics~\cite{PJ-AM-NEF-FB:13d}: 
\begin{align}\label{e8}
  p_{i}(k+1)=\sum_{j=1}^{n}W_{ji}p_{j}(k), 
\end{align} 
where the perception weight $W_{ji}$ coincides exactly with the influence weight in opinion formation. 
By analogy with~\eqref{e8}, a straightforward extension to the FJ setting would be
\begin{align}\label{e8FJ}
p_{i}(k+1)=\sum_{j=1}^{n}a_{j}W_{ji}p_{j}(k)+\frac{1-a_{i}}{n}.
\end{align}
This dynamics removes the relative stubbornness in~\eqref{e6} and reduces to~\eqref{e8} when $A=I_{n}$. However,~\eqref{e8FJ} does \emph{not} converge to the true social power~\eqref{e4}. Therefore,~\eqref{e6} neither reduces to~\eqref{e8} nor can it be obtained by directly incorporating susceptibilities into~\eqref{e8}. Instead, it defines a distinct mechanism with perception weights tuned by relative stubbornness, which is essential for correctly perceiving the FJ social power.

\begin{remark}(Non-invariance of simplex and hypercube) \label{R1}
Unlike the DeGroot perception dynamics~\eqref{e8}, which preserves the probability simplex $\Delta_{n}$, the perception dynamics~\eqref{e6} generally preserves neither $\Delta_{n}$ nor the unit hypercube $[0,1]^n$. Thus, $p(k)$ is not restricted to $\Delta_{n}$ and $[0,1]^n$ in~\eqref{e6}.
\end{remark}
 
For the DeGroot perception dynamics~\eqref{e8}, convergence to the DeGroot social power $\omega\in\mathbb{R}^{n}$ typically relies on consensus-type connectivity assumptions and globally coordinated initializations, i.e., $\lim_{k\to\infty} W^{k}=\mathbf{1}_{n}\omega^{\top}$ and $\mathbf{1}^{\top}_{n}p(0)=1$.
In contrast, the perception dynamics~\eqref{e6} converges to the true FJ social power from any initial perceptions without connectivity assumptions, demonstrating significant robustness to irrational individuals (with negative or excessively large initial perceptions) and disconnected influence networks. 

\begin{proposition}\label{P1}
(Convergence of perception dynamics)
Suppose that Assumption~\ref{A1} holds. Then all trajectories of system~\eqref{e6} starting from $p(0)\in\mathbb{R}^{n}$ converge exponentially to the true social power $x$ of the FJ model given by~\eqref{e4}.
\end{proposition}
\begin{proof}
We write dynamics~\eqref{e6} in the compact form:
\begin{align*}
p(k+1)=(I_{n}-A)W^{\top}A(I_{n}-A)^{-1}p(k)+(I_{n}-A)\frac{\mathbf{1}_{n}}{n}.
\end{align*}
Since $\rho(AW)<1$ under Assumption~\ref{A1}, $I_{n}-(I_{n}-A)W^{\top}A(I_{n}-A)^{-1}$ is non-singular. Hence, we obtain
\begin{align*}
&\lim_{k\to\infty}p(k)\\&=(I_{n}-(I_{n}-A)W^{\top}A(I_{n}-A)^{-1})^{-1}(I_{n}-A)\frac{\mathbf{1}_{n}}{n}\\
&=(I_{n}-A)(I_{n}-W^{\top}A)^{-1}\frac{\mathbf{1}_{n}}{n}=x.
\end{align*}
The exponential convergence rate follows from $\rho((I_{n}-A)W^{\top}A(I_{n}-A)^{-1})<1$. 
\end{proof}


\begin{example}(Static perception dynamics: robustness and non-invariance)\label{ex:baseline}
Consider the perception dynamics~\eqref{e6} with $n=15$, a random row-stochastic
influence matrix $W$, and a susceptibility matrix $A$ ensuring Assumption~\ref{A1}.
In Fig.~\ref{fig1}a, with $a_1=0$ and $10$ independent initial perceptions sampled uniformly from $[-3,3]^n$,
the trajectories of $p_i$, $i=1,2,3$, converge to the corresponding true social power $x_{i}$, illustrating robustness to arbitrary initializations.
In Fig.~\ref{fig1}b, even with $p(0)\in\Delta_n$, neither $p(k)\in[0,1]^n$ nor $\sum_i p_i(k)=1$ holds for $k>0$, 
showing that $\Delta_n$ and $[0,1]^n$ are not positively invariant.

\begin{figure}[t]
  \centerline{\includegraphics[width=0.9\hsize]{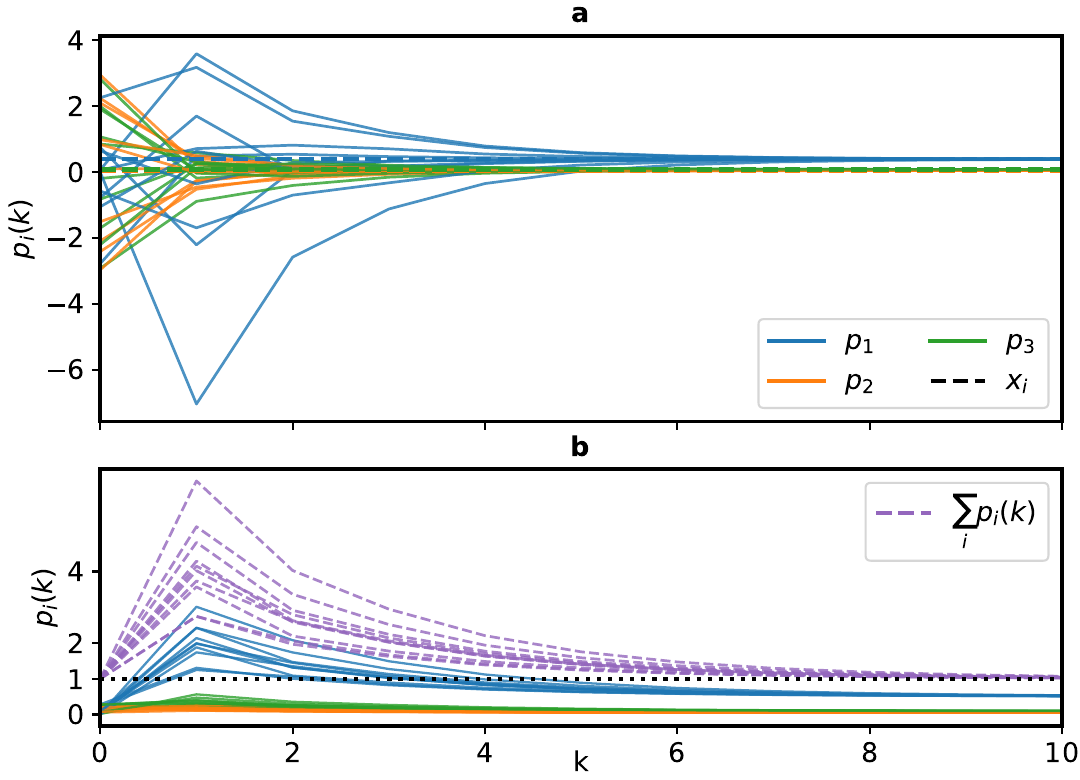}}
  \caption{Trajectories of system~\eqref{e6} under various initializations.}
  \vspace{-1em}
  \label{fig1}
  \end{figure} 
\end{example}


\section{Distributed perception of social power with reflected appraisals}\label{s4}
In Section~\ref{s3}, we studied distributed perception of the static true social power in FJ model under fixed influence networks. In this section, we turn to Problem~\ref{pro2} and investigate whether the perception mechanism remains effective when 
individuals update their self-appraisals using their perceived social power. To this end, we incorporate reflected appraisals into the perception dynamics~\eqref{e6}, as illustrated in Fig.~\ref{fig0}.

\subsection{Distributed perceived-reflected-appraisal dynamics}

Consider the FJ model~\eqref{e2.1} evolving over a sequence of issues. We assume that, on each issue $s$, individuals first update their perceived social power $p(s)$ and then update their self-appraisals using $p(s)$: 
\begin{align}\label{e10}
  W(p(s))=\db{p(s)}+(I_{n}-\db{p(s)})C.
\end{align}
As a result, their perceived social power evolves over issue sequences according to 
\begin{multline}\label{e11}
  p(s+1)=(I_{n}-A)W^{\top}(p(s))A(I_{n}-A)^{-1}p(s)\\+(I_{n}-A)\frac{\mathbf{1}_{n}}{n}. 
\end{multline}

The \emph{perceived-reflected-appraisal} mechanism~\eqref{e11} offers a more realistic explanation of empirically observed reflected-appraisal processes. As shown in Fig.~\ref{fig0}, the true social power dynamics \eqref{e5} operates on the opinion-dynamics layer, which requires global information to compute the true social power. In contrast, dynamics \eqref{e11} evolves on a separate perception layer, where $p$ and $W(p)$ are updated via local interactions. 

We are interested in whether the perceived-reflected-appraisal dynamics~\eqref{e11} correctly perceives the equilibrium of the true social power dynamics~\eqref{e5}.
For the single-timescale case, system~\eqref{e11} can be directly adapted to perceive the true social power governed by \eqref{e5.1} by replacing the timescale $s$ with $k$. Moreover, Lemma~\ref{L1} shows that the true social power dynamics~\eqref{e5} and \eqref{e5.1} converge to the same equilibrium. Therefore, it suffices to focus on the convergence of~\eqref{e11}. 


\subsection{Equilibria: existence, uniqueness and dominance structure}

In this subsection, we analyze the equilibria of dynamics~\eqref{e11}. 
Since true social power is inherently a probability vector, we focus on equilibria of~\eqref{e11} in $\Delta_n$. 

As noted in Remark~\ref{R1}, the underlying perception dynamics~\eqref{e6} preserves neither $\Delta_n$ nor $[0,1]^n$; the same therefore holds for dynamics~\eqref{e11}. This distinguishes it from existing reflected-appraisal dynamics of true social power~\cite{PJ-AM-NEF-FB:13d,YT-PJ-AM-LW-NEF-FB:22} and DeGroot perceived social power~\cite{PJ-NEF-FB:14e,XC-JL-MAB-ZX-TB:16}. 
In particular, the induced influence matrix $W(p(s))$ may contain negative entries when $p_i(s)>1$ for some $i$. As a result, the perception weights are not necessarily nonnegative, and system~\eqref{e11} does not even preserve the nonnegative orthant. Consequently, standard Brouwer-type fixed-point arguments can not be used to establish the existence of equilibria of~\eqref{e11} in $\Delta_n$. 
Despite these challenges, we show that system~\eqref{e11} admits at least one equilibrium in $\Delta_{n}$ by proving that its equilibria coincide with those of the true social power dynamics~\eqref{e5} in $\Delta_{n}$.

\begin{proposition}(Equivalence and existence of equilibria)\label{P2}
Suppose that Assumption~\ref{A1} holds. Then, system~\eqref{e11} shares the same set of equilibria as system~\eqref{e5} and has at least one equilibrium in $\Delta_{n}$. 
\end{proposition} 
\begin{proof}
Assume that $p^{*}\in\Delta_{n}$ is an equilibrium of system~\eqref{e11}. Then, $p^{*}$ satisfies 
\begin{equation*}
p^{*}=(I_{n}-A)W^{\top}(p^{*})A(I_{n}-A)^{-1}p^{*}+\frac{I_{n}-A}{n}\mathbf{1}_{n},
\end{equation*}
that is 
\begin{equation}\label{e3.1}
p^{*}=(I_{n}-A)(I_{n}-W^{\top}(p^{*})A)^{-1}\frac{\mathbf{1}_{n}}{n},
\end{equation}
which means systems~\eqref{e11} and~\eqref{e5} have the same equilibria in $\Delta_{n}$. Hence, existence directly follows from Lemma~\ref{L1}. 
\end{proof}


Proposition~\ref{P2} shows that system~\eqref{e11} preserves the equilibria of the true social power dynamics~\eqref{e5} within $\Delta_n$. By Proposition~\ref{P2}, system~\eqref{e11} admits a unique equilibrium in $\Delta_{n}$ if the condition in Lemma~\ref{L1}~\ref{L1-3} holds. However, Monte Carlo validation suggests that this condition is not necessary~\cite[Conjecture 1]{YT-PJ-AM-LW-NEF-FB:22}.
To go beyond this limitation, we introduce the notions of \emph{partially stubborn paths} and \emph{partially stubborn cycles}, which allow us to capture the interplay between the relative influence network $\mathcal{G}(C)$ and the mapping $\Phi(z)=(I_{n}-AW(z))^{-1}$; see Appendix~\ref{app:aux} for details. 
Exploiting these structural properties, we establish the uniqueness of the equilibrium for system~\eqref{e11} in $\Delta_{n}$. 
 
\begin{theorem}\label{T2}(Uniqueness of equilibrium)
Suppose that Assumption~\ref{A1} holds. 
Then, system~\eqref{e11} has a unique equilibrium $p^{*}$ in $\Delta_{n}$ satisfying $p^{*}\in\interior\Delta_{n}$.
\end{theorem}

Theorem \ref{T2} is proved in Appendix~\ref{app:T2}. 
From Proposition~\ref{P2} and Lemma~\ref{L1}~\ref{L1-1}, Theorem~\ref{T2} also implies that systems~\eqref{e5} and~\eqref{e5.1} admit a unique equilibrium in $\Delta_{n}$, thereby resolving the equilibrium uniqueness conjecture of systems~\eqref{e5} and~\eqref{e5.1} in~\cite{YT-PJ-AM-LW-NEF-FB:22}. 
We hereinafter denote by $p^{*}\in\interior\Delta_{n}$ the unique equilibrium of system~\eqref{e11} in $\Delta_{n}$. Note that if $p_i^*> 1/2$, then $p_i^*>p_j^{*}$ for all $j\neq i$ and $p_i^{*}>\sum_{j\neq i}p_j^*$. Hence, $i$ dominates the group by holding a majority of the power that surpasses the combined power of all others. Next, we provide a necessary condition for the existence of such a dominant individual.


\begin{proposition}\label{P3}(Existence of a dominant individual)
Suppose that Assumption~\ref{A1} holds, and $p^{*}$ is the unique equilibrium of system~\eqref{e11} in $\Delta_{n}$. Then, 
for any $\sigma  \in [1/2,1)$, $i$ is dominant with $p_{i}^{*}>\sigma$ only if 
\begin{align}\label{edi}
\sum_{j=1}^{n}C_{ji}\frac{a_{j}}{1-a_{j}}>\frac{a_{i}}{1-a_{i}}+\frac{n\sigma-1}{n\sigma(1-\sigma)}.
\end{align}
\end{proposition}

Proposition~\ref{P3} is proved in Appendix~\ref{app:T2}. It shows that an individual $i$ can attain dominant social power only if the weighted aggregate of the \emph{susceptibility-stubbornness ratios} of their in-neighbors sufficiently exceed $i$'s own susceptibility-stubbornness ratio. The term  
$\frac{n\sigma-1}{n\sigma(1-\sigma)}$ further quantifies how increasingly demanding this condition becomes as the dominance level $\sigma$ increases. 






\subsection{Invariant sets and convergence}
We have shown that the equilibrium of system~\eqref{e11} coincides with that of true social power dynamics~\eqref{e5} and is unique, i.e., $p^{*}=x^{*}$ in $\Delta_{n}$. 
In this subsection, we analyze the invariance and convergence of system~\eqref{e11}. We identify two classes of positively invariant sets, each characterized by a hyperrectangular structure in $\mathbb{R}^{n}$. On these invariant hyperrectangles, we establish contraction properties of system~\eqref{e11}, which ensure convergence of the perceived social power to $p^{*}$. 

  
\begin{theorem}\label{T6}
Suppose that Assumption~\ref{A1} holds. Define  
\begin{align*}
  b_{i}=\sum_{j\in\Vp}C_{ji}\frac{a_{j}}{1-a_{j}}, 
  \quad d_{i}=\sum_{j\in\Vp}C_{ji}\frac{1+3a_{j}}{4a_{j}}, 
  \end{align*}
for all $i\in\Vp$ and let $\mu, \nu \in\mathbb{R}^{n}$. For system~\eqref{e11},
\begin{enumerate}
  \item \label{T6-1} let $\mathcal{H}=\varGamma_{n} (\mathbf{0}_{n},\nu)$ with $\nu_{i}\geq 1/n+b_{i}/4$ for $i\in\Vf$ and $\nu_{i}=1/2$ for $i\in\Vp$. If 
  \begin{align}\label{equ:15}
    b_{i}\leq \frac{a_{i}}{1-a_{i}}+\frac{2(n-2)}{n}, \quad \forall i\in\Vp,
  \end{align} 
  then $\mathcal{H}$ is positively invariant, and all trajectories starting from $p(0)$ satisfying $p(0)\geq 0$ and $p_{i}(0)\leq 1/2$ for $i\in\Vp$ exponentially converge to $p^*$;
  
  \item \label{T6-2} let $\mathcal{M}=\varGamma_{n} (\mu,\nu)$ with 
  \begin{align*}
    &\mu_{i}\leq\frac{1}{n}-\frac{d_{i}}{4}, \quad
    \nu_{i}\geq \frac{1}{n}+\frac{b_{i}}{4}, \; \forall  i\in\Vf,\\
    & \mu_{i}=-\frac{1-a_{i}}{4a_{i}}, \quad \nu_{i}=\frac{1+a_{i}}{4a_{i}},\; \forall  i\in\Vp.
    \end{align*}
If~\eqref{equ:15} holds, and   
  \begin{align}\label{equ:16}
    d_{i} <\frac{1}{a_{i}}+\frac{4}{n}, \quad \forall i\in\Vp,
    \end{align} 
then $\mathcal{M}$ is positively invariant, and all trajectories starting from $p(0)$ satisfying 
  \begin{align}\label{Mini}
    -\frac{1-a_{i}}{4a_{i}}\leq p_{i}(0)\leq \frac{1+a_{i}}{4a_{i}}, \quad  \forall i\in\Vp,
  \end{align} 
  exponentially converge to $p^*$.
\end{enumerate} 
\end{theorem}

Theorem~\ref{T6} is proved in Appendix~\ref{app:T6}. Since system~\eqref{e11} does not preserve $\Delta_{n}$, invariance-based convergence analysis on $\Delta_{n}$ is not applicable. Instead, Theorem~\ref{T6} identifies two classes of positively invariant hyperrectangles in $\mathbb{R}^{n}$ and establishes contraction results on these sets. In class $\mathcal{H}$, perceptions remain positive: perceptions of partially stubborn individuals are uniformly upper-bounded by $1/2$, whereas perceptions of fully stubborn individuals can be arbitrarily large. In contrast, class $\mathcal{M}$ allows negative perceptions, and these perceptions can be arbitrarily small or large as $a_{i}$ approaches or equals $0$. Since $-\frac{1-a_{i}}{4a_{i}}<0$ and $\frac{1+a_{i}}{4a_{i}}>1/2$, $\mathcal{M}$ allows a broader range of initial perceptions than $\mathcal{H}$.
Theorem~\ref{T6} demonstrates that the underlying perception mechanism~\eqref{e6} remains effective when coupled with reflected appraisal, and that the resulting perceived-reflected-appraisal dynamics~\eqref{e11} retains robustness to extreme initial perceptions. Moreover, in class $\mathcal{M}$ the time-varying influence matrix $W(p(s))$ may have negative entries; the contraction property implies that $W(p(s))$ becomes nonnegative after a finite transient.

Roughly speaking, the contraction condition~\eqref{equ:15} (resp. \eqref{equ:16}) requires a partially stubborn individual to be \emph{less} (resp. \emph{more}) stubborn when it is assigned excessively
large influence by those who are less (resp. more) stubborn. Moreover, by Proposition~\ref{P3},~\eqref{equ:15} precludes partially stubborn individuals from being dominant at equilibrium. 
We illustrate these observations in the following example.

\begin{example}
(Convergence in general case) 
Consider system~\eqref{e11} with $n=3$, $C=[0 \ 0.6 \ 0.4;0 \ 0 \ 1;0.5 \ 0.5 \ 0]$ and $a_1=0$.  Fig.~\ref{fig4}a shows the set of pairs $(a_2,a_3)$ for which $(A,C)$ satisfies~\eqref{equ:15} and~\eqref{equ:16}. This illustrates that the contraction conditions hold for a nontrivial range of susceptibilities. 
Next, consider system~\eqref{e11} with $n=6$. We generate random $C$ and $A$ satisfying~\eqref{equ:15} and~\eqref{equ:16}. 
Fig.~\ref{fig4}b depicts the trajectories of the perceived social power
$p_i(s)$ for $i\in\{4,5,6\}$, where $a_4\approx0.76$, $a_5\approx0.40$, and $a_6\approx0.19$, from $5$ different initial perceptions $p(0)$ satisfying~\eqref{Mini}. The dashed curves show the trajectories of true social power evolution generated by~\eqref{e5} with $x(0)\in\Delta_n$. As shown in Fig.~\ref{fig4}b, 
all sampled trajectories of~\eqref{e11} converge to the equilibrium of true social power dynamics~\eqref{e5}. In particular, for $a_6\approx 0.19$, the admissible interval for $p_{6}(0)$ is approximately $[-1.1, 1.57]$.

\begin{figure}[!t]
    \centerline{\includegraphics[width=\hsize]{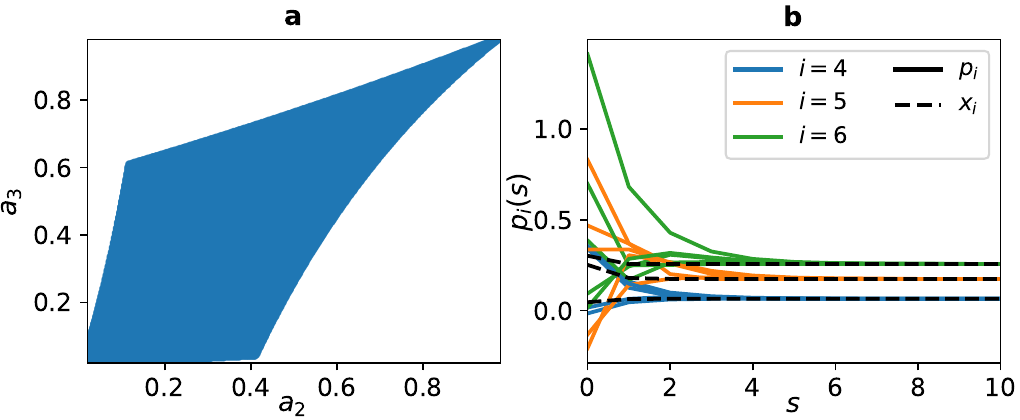}}
    \caption{(a) Feasible region of $(a_2,a_3)$ satisfying~\eqref{equ:15} and \eqref{equ:16} for $n=3$, $a_1=0$, and the given matrix $C$. (b) Trajectories of system~\eqref{e11} from $p(0)$ satisfying ~\eqref{Mini} and system~\eqref{e5} from $x(0)\in\Delta_n$.}
    \vspace{-1em}
    \label{fig4}
    \end{figure}  
\end{example}

By \cite[Theorem 2]{YT-PJ-AM-LW-NEF-FB:22} and Proposition~\ref{P2}, $p^*=\mathbf{1}_{n}/n$, known as the democratic social power, if and only if $A(I_n-A)^{-1}\mathbf{1}_n$ is a left eigenvector of $C$ associated with eigenvalue $1$. This condition is equivalent to $b_i=a_{i}/(1-a_{i})$ for all $i\in\V$, and satisfies \eqref{equ:15}. 
Thus, we obtain the following corollary.  

\begin{corollary}(Convergence to democracy)\label{C0}
Suppose that Assumption~\ref{A1} holds. If $A(I_n-A)^{-1}\mathbf{1}_n$ is a left eigenvector of $C$ associated with eigenvalue $1$, then all trajectories of system~\eqref{e11} starting from $p(0)\in\mathbb{R}^{n}$ satisfying $p(0)\geq 0$ and $p_{i}(0)\leq 1/2$ for $i\in\Vp$ exponentially converge to $\mathbf{1}_n/n$. 
\end{corollary}

In fact, both \eqref{equ:15} and \eqref{equ:16} are satisfied in FJ networks exhibiting either a democratic or a maximally nonuniform true social power structure.
In the next section, we illustrate the applicability of these conditions in structured networks, including star topology and homogeneous stubbornness. 

\section{Convergence in structured networks: Star topology and homogeneous stubbornness}\label{s5}
In influence system theory, star topology captures autocratic influence structure and is of particular interest in the study of social power~\cite{PJ-AM-NEF-FB:13d,YT-PJ-AM-LW-NEF-FB:22}. Moreover, system~\eqref{e11} reduces to a PageRank iteration under homogeneous stubbornness, connecting our model to the classical PageRank algorithm~\cite{DFG:15}. In this section, we verify conditions~\eqref{equ:15} and \eqref{equ:16} for such structured networks and exploit their structure to derive sharper convergence characterizations beyond Theorem~\ref{T6}.

\subsection{Star topology}
We first consider the case where the center node is fully stubborn. In this setting, the true social power at equilibrium exhibits a maximally nonuniform structure~\cite{YT-PJ-AM-LW-NEF-FB:22}, in the sense that  all non-center nodes attain their respective lower bounds while the center node reaches its maximum. Observe that conditions~\eqref{equ:15} and~\eqref{equ:16} hold in this setting. Therefore, Theorem~\ref{T6} yields the following corollary. 
\begin{corollary}(Star topology with fully stubborn center node)\label{C02}
  Suppose that Assumption~\ref{A1} holds and $\mathcal{G}(C)$ is a star topology with a fully stubborn center node. Then, all trajectories of system~\eqref{e11} starting from $p(0)$ satisfying~\eqref{Mini}
  exponentially converge to $p^{*}$.
\end{corollary}

In the next theorem, we exploit the star structure to enlarge the admissible set of initial perceptions beyond that in~\eqref{Mini}.

\begin{theorem}(Star topology with fully stubborn center node)\label{T3}
Suppose that Assumption~\ref{A1} holds and $\mathcal{G}(C)$ is a star topology with a fully stubborn center node. For system \eqref{e11}, 
\begin{enumerate} 
    \item 
    all trajectories starting from $p(0)$ 
    converge to $p^{*}$ if and only if $|p_{i}(0)|<\xi_i$ for all $i\in\Vp$, where $\xi_i=\frac{1+\sqrt{1-4a_{i}(1-a_{i})/n}}{2a_{i}}$; \label{T3-1}
    \item $p^{*}$ is locally exponentially stable.\label{T3-2} 
\end{enumerate}
\end{theorem}

Theorem \ref{T3} is proved in Appendix~\ref{app:T3}. Compared with Corollary~\ref{C02}, it provides an exact characterization of the 
admissible initial perceptions which is less restrictive than~\eqref{Mini}.
In particular, if $|p_{i}(0)|>\xi_{i}$ for some $i\in\Vp$, system \eqref{e11} diverges. However, because 
the center node may exhibit a slow initial transient, we do not obtain a
uniform global exponential rate. Exponential convergence can only be guaranteed after a finite transient or on a more restrictive set of initial
conditions, as shown in Theorem~\ref{T3}\ref{T3-2} and Corollary~\ref{C02}. 

We now study the case where the center node is partially stubborn. Assume that node $1$ is the center node. We focus on the setting where $C_{1j}=0$ for all $j\in\Vp$, under which $p^{*}$ admits an explicit form. Even in this restricted setting, conditions~\eqref{equ:15} and~\eqref{equ:16} are in general not satisfied. 

\begin{theorem}(Star topologies with partially stubborn center node)\label{T4}
Suppose that Assumption~\ref{A1} holds, $\mathcal{G}(C)$ is a star topology with partially stubborn center node $1$, and $C_{1j}=0$ for all $j\in\Vp$. For system \eqref{e11}, if 
\begin{align}\label{equ:17}
  \sum_{j\in\Vp\setminus\{1\}}\frac{a_{j}}{1-a_{j}}\leq\frac{1}{a_{1}(1-a_{1})}-\frac{4}{n},
\end{align} 
\begin{enumerate}
    \item then all trajectories starting from $p(0)$ satisfying 
    \begin{align}\label{equ:17.1}
     |p_{1}(0)|<\frac{1}{2a_{1}} \; \textup{and} \; p_{i}(0)\in [0,1],  \forall i\in \Vp\setminus\{1\},
    \end{align}
     converge to $p^{*}$; \label{T4-2}
    \item $p^{*}$ is exponentially stable. \label{T4-3}
\end{enumerate}
\end{theorem}

Theorem \ref{T4} is proved in Appendix~\ref{app:T3}. Compared with~\eqref{equ:15}, condition~\eqref{equ:17} is milder, while the corresponding admissible set of initial perceptions~\eqref{equ:17.1} is strictly
larger than $\mathcal{H}$.

\subsection{Homogeneous stubbornness and reflected-appraisal PageRank}
Finally, we consider the case of homogeneous stubbornness, i.e., $a_i=a\in(0,1)$ for all $i\in\V$. In this case, system~\eqref{e11} reduces to
\begin{equation}\label{ehomo}
 p(s+1)=a\,W^{\top}(p(s))\,p(s)+\frac{1-a}{n}\mathbf{1}_{n},
\end{equation}
which can be viewed as an extension of the PageRank algorithm~\cite{DFG:15,HI-RT:14}. Unlike the classical PageRank iteration, the link matrix $W(p)$ in~\eqref{ehomo} coevolves with the PageRank vector $p$. In particular, a web page may augment or diminish its outbound links in response to
changes in its perceived importance. This mechanism better reflects the dynamic nature of interactions among web pages. 

For system~\eqref{ehomo}, if the relative influence matrix $C$ is doubly-stochastic, then both conditions~\eqref{equ:15} and~\eqref{equ:16} are satisfied. Hence, Theorem~\ref{T6} yields the following corollary.

\begin{corollary}(Homogeneous stubbornness with doubly-stochastic influence network)\label{C01}
Suppose that $C$ is doubly-stochastic. Then, all trajectories of system~\eqref{ehomo} starting from $p(0)$ satisfying~\eqref{Mini} exponentially converge to $\mathbf{1}_n/n$. 
\end{corollary}

Corollaries~\ref{C02} and~\ref{C01} identify two structured settings in which both conditions~\eqref{equ:15} and~\eqref{equ:16} hold. These cases represent two opposite extremes of social power configurations, namely, the maximally nonuniform and the democratic structures. The corollaries therefore illustrate the applicability of the contraction conditions in Theorem~\ref{T6} to meaningful influence networks spanning both extremes, and show that the proposed dynamics~\eqref{e11} can perceive the true social power in both regimes. Next, we show that system~\eqref{ehomo} preserves the simplex $\Delta_{n}$ and establish exponential convergence on $\Delta_{n}$. 
\begin{theorem}(Homogeneous stubbornness)\label{T5}
For system~\eqref{ehomo},
\begin{enumerate}
\item $\Delta_{n}$ is positively invariant;\label{T5-1}
\item \label{T5-2} all trajectories starting from $p(0)\in \Delta_{n}$ exponentially converge to $p^{*}$ if
\begin{align}\label{equ:19}
a\leq \frac{n-4+\sqrt{17n^{2}-24n+16}}{8(n-1)}.
\end{align} 
\end{enumerate}
\end{theorem}

Theorem~\ref{T5} shows that under homogeneous stubbornness, the perception dynamics preserves $\Delta_n$ and converges exponentially from $\Delta_n$, with a condition independent of the relative influence matrix $C$.

\begin{example}\label{exa3}
(Convergence in structured networks) Consider system~\eqref{e11} with $n=6$. Fig.~\ref{fig3} shows trajectories of the selected nodes $i=1,2$ from $5$ independent initial perceptions under different structured networks. 
In Fig.~\ref{fig3}a, $C$ represents a star topology with $a_1=0$ and $C_{i1}=1$ for all $i\neq 1$. In Fig.~\ref{fig3}b, we choose a random $A$ satisfying~\eqref{equ:17} with $a_1=0.4$ and set $C_{1i}=0$ for all $i$ with $a_i>0$. In both cases, the trajectories converge to the same equilibrium from initial perceptions that go beyond the admissible set given by~\eqref{Mini}. For instance, in Fig.~\ref{fig3}a, with $a_2=0.3$, Theorem~\ref{T3} allows $p_2(0)\in (-3.2,3.2)$, whereas~\eqref{Mini} restricts $p_2(0)\in[-0.58,1.08]$. Figs.~\ref{fig3}a and~\ref{fig3}b also exhibit a slow initial transient, especially for the center node. These observations are consistent with the less restrictive initial conditions and local exponential convergence established in Theorems~\ref{T3} and~\ref{T4}. 
In Fig.~\ref{fig3}c, we set $A=aI_{n}$ with $a=0.6$ and choose a random $C$ such that $\mathcal{G}(C)$ has no globally reachable node. Fig.~\ref{fig3}c shows that this lack of connectivity does not preclude exponential convergence. 
\begin{figure}[!t]
      \centerline{\includegraphics[width=\hsize]{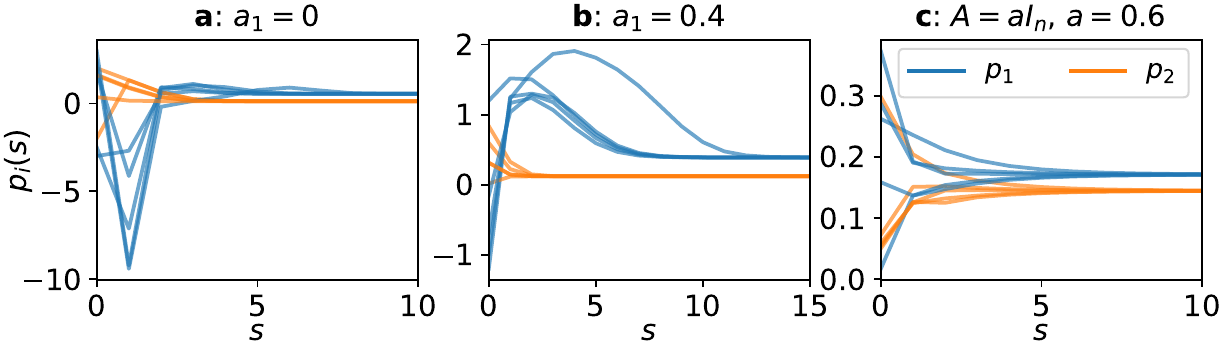}}
      \caption{Trajectories of system~\eqref{e11} under: (a) star topology with fully stubborn center node; (b) star topology with partially stubborn center node; (c) homogeneous stubbornness.}
      \vspace{-1em}
      \label{fig3}
      \end{figure}  
\end{example}

\section{Conclusion}\label{s6}
This paper has investigated distributed perception of social power in groups of stubborn individuals. We proposed a 
distributed perception mechanism based on the FJ model, which enables individuals to estimate their true social power using only local information. We analyzed the 
equilibria, invariant sets and convergence properties of the resulting dynamics in both fixed and reflected-appraisal influence networks. We established conditions under which the perceived social power converges to the true social power.

Our theoretical results show that the proposed dynamics can perceive individuals' social power, effectively and efficiently, under appropriate conditions. The perception dynamics captures that true social power is globally defined and costly to compute, whereas individuals typically only have local information. Particularly, the underlying perception mechanism remains reliable under extreme initial perceptions, disconnected influence networks, reflected-appraisal coevolution, and timescale variations. 
Future work will focus on developing data-driven mechanism of social power perception based on individuals' opinions.

\appendices
\section{Auxiliary lemmas}\label{app:aux}
\setcounter{lemma}{0}  
\renewcommand\thelemma{A.\arabic{lemma}} 
\newtheorem{adefinition}[lemma]{Definition}

\begin{lemma}\label{L1}
(Equilibria and convergence of systems~\eqref{e5} and~\eqref{e5.1}, \cite{YT-PJ-AM-LW-NEF-FB:22})
Suppose that Assumption~\ref{A1} holds. Let $\amax=\max_{i}a_{i}$, $\amin=\min_{i}a_{i}$, and $\zeta =(\sum_{j=1}^{n}a_{j}+1-\amin)/n$. Then, 
\begin{enumerate}
\item systems~\eqref{e5} and~\eqref{e5.1} share the same equilibria in $\Delta_{n}$ and have at least one equilibrium $x^{*}\in\interior\Delta_{n}$; \label{L1-1}
\item if $\amax<1/(1+2\zeta)$ (resp. $\amax<1/2$), then all trajectories of system~\eqref{e5} (resp. system~\eqref{e5.1}) starting from $x(0)\in\Delta_{n}$ exponentially converges to the unique equilibrium $x^{*}\in\interior\Delta_{n}$; \label{L1-3}
\end{enumerate}
\end{lemma}

We have the following definitions for directed paths and cycles consisting of only partially stubborn individuals.  

\begin{adefinition}\label{D2}(Partially stubborn path and partially stubborn cycle)
For a directed path $q^{ij}$ from $i$ to $j$ in $\mathcal{G}(C)$, 
\begin{enumerate}
  \item we call $q^{ij}$ a partially stubborn path (PSP) if it consists of only partially stubborn individuals, i.e., $q^{ij}\setminus\{i,j\}\subset\Vp$;
  \item if $i=j$ for a PSP $q^{ij}$, we call $q^{i}=q^{ij}$ a partially stubborn cycle (PSC);
  \item the weight of a directed path $q^{ij}=(i,l_{1},\dots,l_{r},j)$ with respect to $C$ is defined by $C_{q^{ij}}=C_{il_{1}}C_{l_{1}l_{2}}\dots C_{l_{r}j}$.
\end{enumerate}
\end{adefinition}

Define $\map{\Phi}{[0,1]^{n}}{\mathbb{R}^{n \times n}}$ by $\Phi(z)=(I_{n}-AW(z))^{-1}$. 
The following lemmas capture the properties of $\Phi(z)$.

\begin{lemma}\label{L2}(Properties of $\Phi$)
Suppose that $z\in\interior\Delta_{n}$, and $p^{*}\in\interior\Delta_{n}$ is an equilibrium of system \eqref{e11}, then
\begin{enumerate}
\item $\Phi(z)$ is non-negative with positive diagonal entries, and $\Phi(z)(I_{n}-A)$ is row-stochastic; \label{L2-1}
\item for any partially stubborn individuals $i,j\in\Vp$, $\Phi_{ij}(z)>0$ if and only if there is a PSP $q^{ij}$ in $\mathcal{G}(C)$; \label{L2-2}
\item $\sum_{j=1}^{n}C_{ij}\Phi_{ji}(z)>0$ if and only if there is a PSC $q^{i}$ in $\mathcal{G}(C)$;\label{L2-3}
\item $\Phi_{ii}(z)>\Phi_{ji}(z)$ for any $i\neq j$; \label{L2-4}
\item if there is no PSC of $i\in\Vp$ in $\mathcal{G}(C)$, then $p_{i}^{*}<1/(2a_i)$. \label{L2-5}
\end{enumerate}
\end{lemma}
\begin{proof}
The proofs of \ref{L2-1} and \ref{L2-2} can be found in \cite[Property 2 and Lemma 3]{YT-LW:18}, respectively. \ref{L2-3} is directly implied by~\ref{L2-2}. 

\emph{Proof of~\ref{L2-4}:} Given $i$, let $\Phi_{ji}(x)=\max_{l\neq i}\Phi_{li}(x)$. It is enough to prove $\Phi_{ji}(x)<\Phi_{ii}(x)$. By~\ref{L2-1}, $\Phi_{ii}(x)>0$.  
Suppose that $\Phi_{ji}(x)\geq\Phi_{ii}(x)$. Since $(I_{n}-AW(x))\Phi(x)=I_{n}$ and $W(x)$ is row-stochastic, we obtain 
\begin{align*}
\Phi_{ji}(x)=a_{j}\sum_{l=1}^{n}W_{jl}(x)\Phi_{li}(x)\leq a_{j}\Phi_{ji}(x),
\end{align*}
which contradicts $a_{j}<1$. Therefore, $\Phi_{ii}(x)>\max_{l\neq i}\Phi_{li}(x)$.

\emph{Proof of~\ref{L2-5}:} Since~\ref{L2-5} is trival if $a_{i}\leq 1/2$, we assume $a_{i}> 1/2$. From \ref{L2-1} and \ref{L2-3} we have $\sum_{j=1}^{n}C_{ij}\Phi_{ji}(p^{*})=0$. Moreover, \eqref{e3.1} implies 
\begin{align*}
  p_{i}^{*}=\frac{1-a_{i}}{n}\sum_{j=1}^{n}\Phi_{ji}(p^{*})<(1-a_{i})\Phi_{ii}(p^{*})=\frac{1-a_{i}}{1-a_{i}p_{i}^{*}},
\end{align*}
where the inequality is implied by statement \ref{L2-4}, and the last equality follows from $(I_{n}-AW(p^{*}))\Phi(p^{*})=I_{n}$ and $\sum_{j=1}^{n}C_{ij}\Phi_{ji}(p^{*})=0$. Hence, 
\begin{align*}
  p_{i}^{*}(1-a_{i})+a_{i}p_{i}^{*}(1-p_{i}^{*})=p_{i}^{*}(1-a_{i}p_{i}^{*})<1-a_{i}, 
\end{align*}
which yields $a_{i}p_{i}^{*}<1-a_{i}<1/2$. Thus, $p_{i}^{*}<1/(2a_{i})$.
\end{proof}

\begin{lemma}(Connections between $\Phi$ and PSCs)\label{L3}
  Suppose that node $i$ has PSCs in $\mathcal{G}(C)$. Let $\mathcal{C}_{i}=\{q^{i}_{1},\dots,q^{i}_{m_{i}}\}$ be the set of all PSCs of $i$, where $m_{i}=|\mathcal{C}_{i}|$. Define $\eta_{i}$ and $\map{\phi_{i}}{\Delta_{n}}{\mathbb{R}}$ as
  \begin{align*}
    \eta_{i}=\frac{a_{i}(1-x_{i})}{1-a_{i}x_{i}}, \quad \textup{and} \quad
    \phi_{i}(x)=\sum_{j=1}^{m_{i}}C_{q^{i}_{j}}\prod_{l\in q^{i}_{j}\setminus\{i\}}\eta_{l}. 
    \end{align*}  
\begin{enumerate}
  \item For any $x,z\in\interior\Delta_{n}$,  {\small\begin{align*}
  \Phi_{ii}(x)=\frac{1}{1-a_{i}x_{i}-a_{i}(1-x_{i})\phi_{i}(x)};
  \end{align*}} \label{L3-1}
  \item if $x_{j}>z_{j}$ for all $j\neq i$, then \label{L3-2}
  \begin{align*}
  \phi_{i}(x)\geq\phi_{i}(z)\prod_{j\neq i}\frac{1-x_{j}}{1-z_{j}}\frac{1-a_{j}z_{j}}{1-a_{j}x_{j}}.
  \end{align*}
\end{enumerate}
\end{lemma}

\begin{proof} By $(I_{n}-AW(x))\Phi(x)=I_{n}$, we obtain
  \begin{equation}\label{equa:Phi_ii}
    \begin{cases}
      \Phi_{ii}(x)=\frac{1}{1-a_{i}x_{i}}+\eta_{i}\sum_{l=1}^{n}C_{il}\Phi_{li}(x),\\
      \Phi_{ji}(x)=\eta_{j}\sum_{l=1}^{n}C_{jl}\Phi_{li}(x).
    \end{cases}
    \end{equation}

\emph{Proof of~\ref{L3-1}:} 
If $i\in\Vf$, then $\eta_{i}=a_{i}=0$ and $\Phi_{ii}(x)=1$, which is a trivial case. We now focus on the case that $i\in\Vp$. Note that if $j\in\Vf$, then $\Phi_{ji}(x)=\eta_{j}=0$. For any $j\in\Vp\cap \mathcal{N}_{i}^{-}$ with $\Phi_{ji}(x)>0$, by Lemma \ref{L2}~\ref{L2-2}, assume that $\mathcal{C}_{ji}=\{q_{1}^{ji},\dots, q_{m_{ji}}^{ji}\}$ is the set of all the PSPs from $j$ to $i$ in $\mathcal{G}(C)$. Then, we have 
\begin{align}\label{equa:PSP}
    \Phi_{ji}(x)=&\eta_{j}\sum_{l^{1}\in\Vp\cap\mathcal{N}_{j}^{-}}C_{jl^{1}}\Phi_{l^{1}i}(x) \notag\\
    =&\eta_{j}\sum_{l^{1}\in\Vp\cap\mathcal{N}_{j}^{-}}C_{jl^{1}}\eta_{l^{1}}\sum_{l^{2}\in\Vp\cap\mathcal{N}_{l^{1}}^{-}}C_{l^{1}l^{2}}\Phi_{l^{2}i}(x) \notag\\
    =&\sum_{l^{1}\in\Vp\cap\mathcal{N}_{j}^{-}}\sum_{l^{2}\in\Vp\cap\mathcal{N}_{l^{1}}^{-}}\dots\sum_{l\in\Vp\cap\mathcal{N}_{i}^{+}}\eta_{j}\eta_{l^{1}}\eta_{l^{2}}\dots\eta_{l}\notag\\
      &C_{jl^{1}}C_{l^{1}l^{2}}\dots C_{li}\Phi_{ii}(x)\\
  =&\Phi_{ii}(x)\sum_{h=1}^{m_{ji}}C_{q_{h}^{ji}}\prod_{l\in q_{h}^{ji}\setminus\{i\}}\eta_{l}, \notag
\end{align}
where each sequence $(j,l^1,\dots,l,i)$ is a PSP in $\mathcal{C}_{ji}$. Combining \eqref{equa:Phi_ii} and \eqref{equa:PSP}, we obtain 
\begin{multline*}
  \Phi_{ii}(x)=\\\frac{1}{1-a_{i}x_{i}}+\eta_{i}\Phi_{ii}(x)\sum_{j\in\Vp\cap \mathcal{N}_{i}^{-}}C_{ij}\sum_{h=1}^{m_{ji}}C_{p_{h}^{ji}}\prod_{l\in p_{h}^{ji}\setminus\{i\}}\eta_{l}\\
  =\frac{1}{1-a_{i}x_{i}}+\frac{a_{i}(1-x_{i})}{1-a_{i}x_{i}}\Phi_{ii}(x)\phi_{i}(x),
\end{multline*}
which yields 
\begin{align*}
  \Phi_{ii}(x)=\frac{1}{1-a_{i}x_{i}-a_{i}(1-x_{i})\phi_{i}(x)}.
\end{align*}

\emph{Proof of~\ref{L3-2}:} For all $j\neq i$, since $x_{j}>z_{j}$, we have 
\begin{align*}
  \frac{1-x_{j}}{1-z_{j}}\frac{1-a_{j}z_{j}}{1-a_{j}x_{j}}<1.
\end{align*}
Therefore, by $q^{i}_{j}\subset\V$ for all $j\in\until{m_{i}}$, we obtain 
\begin{align*}
  \phi_{i}(x)&=\sum_{j=1}^{m_{i}}C_{q^{i}_{j}}\prod_{l\in q^{i}_{j}\setminus\{i\}}\frac{a_{l}(1-x_{l})}{1-a_{l}x_{l}}\\
  &=\sum_{j=1}^{m_{i}}C_{q^{i}_{j}}\prod_{l\in q^{i}_{j}\setminus\{i\}}\frac{a_{l}(1-z_{l})}{1-a_{l}z_{l}}\frac{1-x_{l}}{1-z_{l}}\frac{1-a_{l}z_{l}}{1-a_{l}x_{l}}\\
  &\geq\prod_{h\neq i}\frac{1-x_{h}}{1-z_{h}}\frac{1-a_{h}z_{h}}{1-a_{h}x_{h}}\sum_{j=1}^{m_{i}}C_{q^{i}_{j}}\prod_{l\in q^{i}_{j}\setminus\{i\}}\frac{a_{l}(1-z_{l})}{1-a_{l}z_{l}}\\
  &=\phi_{i}(z)\prod_{h\neq i}\frac{1-x_{h}}{1-z_{h}}\frac{1-a_{h}z_{h}}{1-a_{h}x_{h}}.
\end{align*}
\end{proof}

\begin{lemma}\label{L4}
Suppose $0<m<n$ and $a_{i}\in (0,1)$ for all $i\in\until{m}$. If $\sum_{j=2}^{m}\frac{a_{j}}{1-a_{j}}\leq\frac{1}{a_{1}(1-a_{1})}-\frac{4}{n}$, then 
\begin{align*}
  \sum_{j=2}^{m}a_{j}< \frac{n}{4a_{1}(1-a_{1})}-\frac{1-a_{1}}{16a_{1}}n-1.
\end{align*}
\end{lemma}
\begin{proof}
By $\sum_{j=2}^{m}\frac{a_{j}}{1-a_{j}}\leq\frac{1}{a_{1}(1-a_{1})}-\frac{4}{n}$ we have
\begin{align*}
  \sum_{j=2}^{m}\frac{1}{1-a_{j}}=\sum_{j=2}^{m}(\frac{a_{j}}{1-a_{j}}+1)\leq \frac{1}{a_{1}(1-a_{1})}-\frac{4}{n}+m-1.
\end{align*}
Then, Cauchy's inequality~\cite[Theorem 7]{GHH-JEL-GP:34} yields 
\begin{align*}
  \sum_{j=2}^{m}(1-a_{j})\geq\frac{(m-1)^{2}}{\sum_{j=2}^{m}\frac{1}{1-a_{j}}}\geq \frac{(m-1)^{2}}{\frac{1}{a_{1}(1-a_{1})}-\frac{4}{n}+m-1},
\end{align*}
which implies 
\begin{align*}
  \sum_{j=2}^{m}a_{j}\leq m-1-\frac{(m-1)^{2}}{\frac{1}{a_{1}(1-a_{1})}-\frac{4}{n}+m-1}.
\end{align*}
Hence, it suffices to prove 
\begin{align*}
m-\frac{(m-1)^{2}}{\frac{1}{a_{1}(1-a_{1})}-\frac{4}{n}+m-1}<\frac{n}{4a_{1}(1-a_{1})}-\frac{1-a_{1}}{16a_{1}}n,
\end{align*}
which is equivalent to $(m-1)g_{1}/n< g_{2}$
with $g_{1}=a_{1}(1-a_{1})(n(1-a_{1})^{2}-4h(n-4))$, $g_{2}=4h^{2}-(1-a_{1})^{2}h$ and $h=1-4a_{1}(1-a_{1})/n$. By $0<a_{1}(1-a_{1})\leq 1/4$, we have $1/2\leq 1-1/n\leq h<1$, which means 
\begin{align*}
  g_{1}<a_{1}(1-a_{1})(20-3n-\frac{16}{n})\leq 6a_{1}(1-a_{1}),
\end{align*}
as $20-3n-16/n$ strictly decreases with respect to $n$. Thus, 
\begin{align*}
\frac{m-1}{n}g_{1}<6a_{1}(1-a_{1})h\leq (2-(1-a_{1})^{2})h<g_{2},
\end{align*}
where the last inequality is implied by $2\leq 4h$.
\end{proof}
\section{Proofs of Theorem~\ref{T2} and Proposition~\ref{P3}}\label{app:T2}
\subsection{Proof of Theorem~\ref{T2}}
Suppose that $p^{*},\hat{p}^{*}\in\interior\Delta_{n}$ are both equilibria of system \eqref{e11}. By \eqref{e3.1}, we obtain 
\begin{multline}\label{equ:18}
(I_{n}-W^{\top}(p^{*})A)(I_{n}-A)^{-1}p^{*}\\-(I_{n}-W^{\top}(\hat{p}^{*})A)(I_{n}-A)^{-1}\hat{p}^{*}=0.
\end{multline}
Since $\db{p^{*}}p^{*}-\db{\hat{p}^{*}}\hat{p}^{*}=\db{p^{*}+\hat{p}^{*}}(p^{*}-\hat{p}^{*})$, we have 
\begin{align*}
&W^{\top}(p^{*})A(I_{n}-A)^{-1}p^{*}-W^{\top}(\hat{p}^{*})A(I_{n}-A)^{-1}\hat{p}^{*}\\=& \db{p^{*}+\hat{p}^{*}}A(I_{n}-A)^{-1}(p^{*}-\hat{p}^{*})\\&+C^{\top}(I_{n}-\db{p^{*}+\hat{p}^{*}})A(I_{n}-A)^{-1}(p^{*}-\hat{p}^{*})\\=&W^{\top}(p^{*}+\hat{p}^{*})A(I_{n}-A)^{-1}(p^{*}-\hat{p}^{*}).
\end{align*}
Hence, \eqref{equ:18} can be rearranged as:  
\begin{equation*}
(I_{n}-W^{\top}(p^{*}+\hat{p}^{*})A)(I_{n}-A)^{-1}(p^{*}-\hat{p}^{*})=0,
\end{equation*}
which implies $p^{*}=\hat{p}^{*}$ if and only if $I_{n}-AW(p^{*}+\hat{p}^{*})$ is non-singular. If $p^{*}+\hat{p}^{*}\leq\mathbf{1}_{n}$, then $W(p^{*}+\hat{p}^{*})$ is row-stochastic and $\rho(AW(p^{*}+\hat{p}^{*}))<1$ under Assumption \ref{A1}, which means $I_{n}-AW(p^{*}+\hat{p}^{*})$ is non-singular.

We now turn to the case that $p^{*}+\hat{p}^{*}\nleq\mathbf{1}_{n}$. Without loss of generality, let $p_{1}^{*}+\hat{p}_{1}^{*}>1$ with $p_{1}^{*}\geq\hat{p}_{1}^{*}$. Then, $p^{*},\hat{p}^{*}\in\interior\Delta_{n}$ necessitates $p_{1}^{*}>1/2>\sum_{i\neq 1}p_{i}^{*}$ and $0<p_{i}^{*}+\hat{p}_{i}^{*}<1$ for all $i\neq 1$.
Let $z=p^{*}+\hat{p}^{*}-\mathbf{e}_{1}$, where $\mathbf{e}_{1}$ is the $1$-th standard basis. Then, $z\in\interior\Delta_{n}$ and
\begin{align*}
I_{n}-AW(p^{*}+\hat{p}^{*})
=&I_{n}-AW(z)-A\db{\mathbf{e}_{1}}(I_{n}-C)\\
=&I_{n}-AW(z)-a_{1}\mathbf{e}_{1}(\mathbf{e}_{1}-C^{\top}\mathbf{e}_{1})^{\top},
\end{align*}
By Cauchy's formula for rank-one permutation of determinant~\cite[equation (0.8.5.11)]{RAH-CRJ:12},
 \begin{align*}
&\det(I_{n}-AW(p^{*}+\hat{p}^{*}))\\=&\det(I_{n}-AW(z))-a_{1}(\mathbf{e}_{1}-C^{\top}\mathbf{e}_{1})^{\top}\adj(I_{n}-AW(z))\mathbf{e}_{1}\\
=&\det(I_{n}-AW(z))(1-a_{1}\mathbf{e}^{\top}_{1}(I_{n}-C)\Phi(z)\mathbf{e}_{1}),
\end{align*}
where $\det(I_{n}-AW(z))\neq 0$ is implied by $z\in\interior\Delta_{n}$ and Assumption \ref{A1}. Thus, $I_{n}-AW(p^{*}+\hat{p}^{*})$ is non-singular if and only if $1-a_{1}\mathbf{e}^{\top}_{1}(I_{n}-C)\Phi(z)\mathbf{e}_{1}\neq 0$. Since this is trivial for $a_{1}=0$, we assume $a_{1}>0$, and we shall prove $\mathbf{e}^{\top}_{1}(I_{n}-C)\Phi(z)\mathbf{e}_{1}< 1/a_{1}$. Specifically, we consider two cases: $a_{1}z_{1}\leq 1-a_{1}$ and $a_{1}z_{1}> 1-a_{1}$. 

\emph{Case 1. $a_{1}z_{1}\leq 1-a_{1}$:}  
By \eqref{equa:Phi_ii}, we obtain
\begin{multline}\label{e14}
  \mathbf{e}^{\top}_{1}(I_{n}-C)\Phi(z)\mathbf{e}_{1}=\Phi_{11}(z)-\sum_{j=1}^{n}C_{1j}\Phi_{j1}(z)\\=\frac{1}{1-a_{1}z_{1}}-\frac{1-a_{1}}{1-a_{1}z_{1}}\sum_{j=1}^{n}C_{1j}\Phi_{j1}(z).
\end{multline}
Hence, $\mathbf{e}^{\top}_{1}(I_{n}-C)\Phi(z)\mathbf{e}_{1}< 1/a_{1}$ if and only if  
\begin{equation}\label{e14.1}
  a_{1}z_{1}<1-a_{1}+a_{1}(1-a_{1})\sum_{j=1}^{n}C_{1j}\Phi_{j1}(z). 
\end{equation}
If $\sum_{j=1}^{n}C_{1j}\Phi_{j1}(z)>0$, then \eqref{e14.1} holds. Otherwise, by Lemma~\ref{L2}~\ref{L2-2}, $\sum_{j=1}^{n}C_{1j}\Phi_{j1}(z)=0$ is equivalent to $\sum_{j=1}^{n}C_{1j}\Phi_{j1}(p^{*})=0$ and $\sum_{j=1}^{n}C_{1j}\Phi_{j1}(\hat{p}^{*})=0$. Hence, by Lemma~\ref{L2}~\ref{L2-5}, we have $\hat{p}_{1}^{*}\leq p_{1}^{*}<1/(2a_{1})$. As a result, $a_{1}z_{1}=a_{1}(p_{1}^{*}+\hat{p}_{1}^{*})-a_{1}<1-a_{1}$ and \eqref{e14.1} holds.

\emph{Case 2. $a_{1}z_{1}>1-a_{1}$:}
Combining \eqref{equa:Phi_ii} and \eqref{e14.1}, $\mathbf{e}^{\top}_{1}(I_{n}-C)\Phi(z)\mathbf{e}_{1}< 1/a_{1}$ is equivalent to $\Phi_{11}(z)>z_{1}/(1-a_{1})$. Since $a_{1}z_{1}>1-a_{1}$, we obtain $(1-a_{1})/a_{1}<z_{1}<\hat{p}^{*}_{1}\leq p^{*}_{1}$, $p_{1}^{*}>1/(2a_{1})$ and $a_{1}>1/2$. Hence, Lemma~\ref{L2}~\ref{L2-5} implies that node $1$ has PSCs in $\mathcal{G}(C)$. By Lemma~\ref{L3}~\ref{L3-1}, 
$\Phi_{11}(z)>z_{1}/(1-a_{1})$ if and only if 
\begin{align}\label{e15.1}
  \phi_{1}(z)>\frac{a_{1}z_{1}-(1-a_{1})}{a_{1}z_{1}}.
\end{align}
By~\eqref{e3.1}, Lemma~\ref{L2}~\ref{L2-4} and Lemma~\ref{L3}~\ref{L3-1}, we obtain 
\begin{align*}
\hat{p}^{*}_{1}&=\frac{1-a_{1}}{n}\sum_{j=1}^{n}\Phi_{j1}(\hat{p}^{*})<(1-a_{1})\Phi_{11}(\hat{p}^{*})\\ &=\frac{1-a_{1}}{1-a_{1}\hat{p}^{*}_{1}-a_{1}(1-\hat{p}^{*}_{1})\phi_{1}(\hat{p}^{*})},
\end{align*}
which further yields
\begin{align}\label{e15.2}
  \phi_{1}(\hat{p}^{*})>\frac{a_{1}\hat{p}^{*}_{1}-(1-a_{1})}{a_{1}\hat{p}^{*}_{1}}.
\end{align}
Moreover, since $z_{j}>\hat{p}^{*}_{j}$ for all $j\neq 1$, Lemma~\ref{L3}~\ref{L3-2} implies 
\begin{align*}
  \phi_{1}(z)&\geq\phi_{1}(\hat{p}^{*})\prod_{l\neq 1}\frac{1-z_{l}}{1-\hat{p}^{*}_{l}}\frac{1-a_{l}\hat{p}^{*}_{l}}{1-a_{l}z_{l}}\\
  &>\frac{a_{1}\hat{p}^{*}_{1}-(1-a_{1})}{a_{1}\hat{p}^{*}_{1}}\prod_{l\neq 1}\frac{1-z_{l}}{1-\hat{p}^{*}_{l}}\frac{1-a_{l}\hat{p}^{*}_{l}}{1-a_{l}z_{l}}.
  \end{align*}
To ensure \eqref{e15.1}, we only need to prove 
\begin{align}\label{e15}
  \prod_{l\neq 1}\frac{1-z_{l}}{1-\hat{p}^{*}_{l}}\frac{1-a_{l}\hat{p}^{*}_{l}}{1-a_{l}z_{l}}\geq\frac{a_{1}z_{1}-(1-a_{1})}{a_{1}\hat{p}^{*}_{1}-(1-a_{1})}\frac{\hat{p}^{*}_{1}}{z_{1}}.
\end{align}
Since $z_{l}=p_{l}^{*}+\hat{p}_{l}^{*}<1$ for all $l\neq 1$, we have 
\begin{align*}
  \prod_{l\neq 1}\frac{1-z_{l}}{1-\hat{p}^{*}_{l}}\frac{1-a_{l}\hat{p}^{*}_{l}}{1-a_{l}z_{l}}&=\prod_{l\neq 1}(1-\frac{(1-a_{l})p^{*}_{l}}{(1-\hat{p}^{*}_{l})(1-a_{l}z_{l})}).
\end{align*}
By $1-\hat{p}^{*}_{l}>p^{*}_{l}$, we have $(1-a_{l})p^{*}_{l}<(1-a_{l})(1-\hat{p}^{*}_{l})+a_{l}(1-z_{l})(1-\hat{p}^{*}_{l})=(1-\hat{p}^{*}_{l})(1-a_{l}z_{l})$, that is, $0<\frac{(1-a_{l})p^{*}_{l}}{(1-\hat{p}^{*}_{l})(1-a_{l}z_{l})}<1$. By Bernoulli's inequality~\cite[Theorem 58]{GHH-JEL-GP:34}, we obtain 
\begin{align*}
\prod_{l\neq 1}(1-\frac{(1-a_{l})p^{*}_{l}}{(1-\hat{p}^{*}_{l})(1-a_{l}z_{l})})&\geq 1-\sum_{l\neq 1}\frac{(1-a_{l})p^{*}_{l}}{(1-\hat{p}^{*}_{l})(1-a_{l}z_{l})}.
\end{align*}
On the other hand, by $z_{1}=p_{1}^{*}+\hat{p}_{1}^{*}-1$ we obtain 
\begin{align*}
  \frac{a_{1}z_{1}-(1-a_{1})}{a_{1}\hat{p}^{*}_{1}-(1-a_{1})}\frac{\hat{p}^{*}_{1}}{z_{1}}=1-\sum_{l\neq 1}\frac{(1-a_{1})p^{*}_{l}}{(a_{1}\hat{p}^{*}_{1}-(1-a_{1}))z_{1}}.
\end{align*}
Therefore, \eqref{e15} holds if for all $l\neq 1$, 
\begin{align*}
  \frac{1-a_{1}}{(a_{1}\hat{p}^{*}_{1}-(1-a_{1}))z_{1}}\geq\frac{1-a_{l}}{(1-\hat{p}^{*}_{l})(1-a_{l}z_{l})}.
\end{align*}
Since $1-a_{l}\leq 1-a_{l}z_{l}$ and $1-\hat{p}^{*}_{l}\geq\hat{p}^{*}_{1}$, we have 
\begin{align*}
  \frac{(1-a_{l})}{(1-\hat{p}^{*}_{l})(1-a_{l}z_{l})}\leq\frac{1}{1-\hat{p}^{*}_{l}}\leq\frac{1}{\hat{p}^{*}_{1}}.
\end{align*}
It remains to show
\begin{align*}
  \frac{1}{\hat{p}^{*}_{1}}\leq\frac{1-a_{1}}{(a_{1}\hat{p}^{*}_{1}-(1-a_{1}))z_{1}}.
\end{align*}
By $z_{1}=\hat{p}^{*}_{1}+p^{*}_{1}-1$, this is equivalent to 
\begin{align*}
  a_{1}(\hat{p}^{*}_{1})^{2}-(a_{1}p^{*}_{1}+2(1-a_{1}))\hat{p}^{*}_{1}+(1-a_{1})(1-p^{*}_{1})\leq 0,
\end{align*}
which directly follows from  
$(1-a_{1})/a_{1}<\hat{p}^{*}_{1}\leq p^{*}_{1}$ and $p^{*}_{1}>1/(2a_{1})$.


\subsection{Proof of Proposition~\ref{P3}}

By~\eqref{e11}, we have 
\begin{align*}
p_{i}^{*}=a_{i}(p_{i}^{*})^{2}+(1-a_{i})\sum_{j=1}^{n}C_{ji}\frac{a_{j}}{1-a_{j}}p_{j}^{*}(1-p_{j}^{*})+\frac{1-a_{i}}{n},
\end{align*}
which is equivalent to 
\begin{align*}
p_{i}^{*}=\sum_{j=1}^{n}C_{ji}\frac{a_{j}}{1-a_{j}}p_{j}^{*}(1-p_{j}^{*})-\frac{a_{i}}{1-a_{i}}p_{i}^{*}(1-p_{i}^{*})+\frac{1}{n}.
\end{align*}
Since $p_{i}^{*}>1/2$, we obtain $p_{i}^{*}(1-p_{i}^{*})\geq p_{j}^{*}(1-p_{j}^{*})$ for all $j\neq i$. Thus, 
\begin{align}\label{e14.5}
p_{i}^{*}\leq (\sum_{j=1}^{n}C_{ji}\frac{a_{j}}{1-a_{j}}-\frac{a_{i}}{1-a_{i}})p_{i}^{*}(1-p_{i}^{*})+\frac{1}{n}.
\end{align}
Let $\theta=\sum_{j=1}^{n}C_{ji}\frac{a_{j}}{1-a_{j}}-\frac{a_{i}}{1-a_{i}}$, then $\theta>0$ follows from 
\begin{align*}
\theta p_{i}^{*}(1-p_{i}^{*})\geq p_{i}^{*}-\frac{1}{n}>\frac{1}{2}-\frac{1}{n}\geq 0.
\end{align*}
Moreover,~\eqref{e14.5} implies 
\begin{align*}
p_{i}^{*}\leq\frac{\theta-1+\sqrt{(\theta-1)^{2}+4\theta/n}}{2\theta},
\end{align*}
which, combined with $p_{i}^{*}>\sigma$, yields $\theta>\frac{n\sigma-1}{n\sigma(1-\sigma)}$. 

\section{Proof of Theorem \ref{T6}}\label{app:T6}
We first prove~\ref{T6-2} in detail.  
The proof of~\ref{T6-1} follows the same steps with simplified bounds and is provided below.

\emph{Invariance:} Let $p(s)\in\mathcal{M}$. By \eqref{e11}, we have
\begin{align*}
p_{i}(s+1)\leq a_{i}\frac{(1+a_{i})^{2}}{16a_{i}^{2}}+\frac{1-a_{i}}{4}b_{i}+\frac{1-a_{i}}{n}<\frac{1+a_{i}}{4a_{i}}
\end{align*}
for any $i\in\Vp$, where the first inequality is implied by $p_{j}(s)\leq\frac{1+a_{j}}{4a_{j}}$ and $p_{j}(s)(1-p_{j}(s))\leq 1/4$ for all $j\in\Vp$, the second inequality follows from \eqref{equ:15}. On the other hand, we obtain 
\begin{align*}
  &p_{i}(s+1)\\\geq&-(1-a_{i})\sum_{j\in\Vp}C_{ji}\frac{a_{j}}{1-a_{j}}\frac{(1-a_{j})(1+3a_{j})}{16a_{j}^{2}}+\frac{1-a_{i}}{n}\\
  = & -\frac{1-a_{i}}{4}d_{i}+\frac{1-a_{i}}{n}> -\frac{1-a_{i}}{4a_{i}},
  \end{align*}
where the first inequality is implied by $p_{j}(s)(1-p_{j}(s))\geq-\frac{(1-a_{j})(1+3a_{j})}{16a_{j}^{2}}$ for $p_{j}(s)\in [-\frac{1-a_{j}}{4a_{j}},\frac{1+a_{j}}{4a_{j}}]$, and the last inequality follows from \eqref{equ:16}.
Hence, $p_{i}(s+1)\in (-\frac{1-a_{i}}{4a_{i}},\frac{1+a_{i}}{4a_{i}})$ for all $i\in\Vp$. Similarly, for any $i\in\Vf$, we have 
\begin{align*}
  \mu_{i}\leq\frac{1}{n}-\frac{1}{4}d_{i}\leq p_{i}(s+1)\leq\frac{1}{n}+\frac{1}{4}b_{i} \leq \nu_{i}
  \end{align*}
Therefore, $p(s+1)\in\mathcal{M}$, and $\mathcal{M}$ is positively invariant. 

\emph{Convergence:} Define $\map{F}{\mathcal{M}}{\mathcal{M}}$ by $F(x)=(I_{n}-A)W^{\top}(x)A(I_{n}-A)^{-1}x+(I_{n}-A)\mathbf{1}_{n}/n$.
Then, $F$ is differentiable on $\interior\mathcal{M}$ and continuous on $\mathcal{M}$ with $p(s+1)=F(p(s))$. For an infinitesimal displacement $\delta p(s)$ of $p(s)$, we have $\delta p(s+1)=\frac{\partial F}{\partial x}(p(s))\delta p(s)$, where 
{\small
\begin{align*}
  \frac{\partial F}{\partial x}(p(s))\!=\!(I_n-A)\left(\db{2p(s)}\!+\!C^{\top}(I_{n}\!-\!\db{2p(s)})\right)A(I_n-A)^{-1}
\end{align*}
}
is the Jacobian of $F$ at $p(s)$. Define the transformed system $\delta \tilde{p}(s)=(I_{n}-A)^{-1}\delta p(s)$. Then, $\delta \tilde{p}(s+1)=J(p(s))\delta \tilde{p}(s)$ with 
\begin{equation*}
\resizebox{0.48\textwidth}{!}{$\begin{aligned}
   &J(p(s))=(I_{n}-A)^{-1}\frac{\partial F}{\partial x}(p(s))(I_{n}-A)=\\&\begin{bmatrix}
    2a_{1}p_{1}(s) & C_{21}a_{2}(1-2p_{2}(s)) & \dots & C_{n1}a_{n}(1-2p_{n}(s))\\
    C_{12}a_{1}(1-2p_{1}(s)) & 2a_{2}p_{2}(s) & \dots & C_{n2}a_{n}(1-2p_{n}(s))\\
    \vdots & \vdots & \vdots & \vdots\\
    C_{1n}a_{1}(1-2p_{1}(s)) & C_{2n}a_{2}(1-2p_{2}(s)) & \dots & 2a_{n}p_{n}(s)
 \end{bmatrix}.
 \end{aligned}$} 
\end{equation*}
Recall in the proof invariance, we prove that $p_{i}(s)\in (-\frac{1-a_{i}}{4a_{i}},\frac{1+a_{i}}{4a_{i}})$ for all $i\in\Vp$ and $s>0$. Therefore, 
\begin{align*}
  \norm{J(p(s))}_{1}=\max_{i\in\Vp}a_{i}(2| p_{i}(s)|+| 1-2p_{i}(s) |)<1.
\end{align*}
By \cite[Definition 2 and Theorem 2]{MY-JL-BDOA-CY-TB:17b}, 
$\mathcal{M}$ is a generalized contraction region and all trajectories starting from $p(0)\in\mathcal{M}$ exponentially converge to a unique equilibrium in $\mathcal{M}$. In addition, by Proposition \ref{P3} and \eqref{equ:15}, $p_{i}^{*}\leq 1/2<\frac{1+a_{i}}{4a_{i}}$ for all $i\in\Vp$, which further implies 
\begin{align*}
  p_{i}^{*}=\frac{1}{n}+\sum_{j\in\Vp}\frac{a_{j}}{1-a_{j}}C_{ji}p_{j}^{*}(1-p_{j}^{*})\leq \frac{1}{n}+\frac{1}{4}b_{i}\leq \nu_i
\end{align*}
for all $i\in\Vf$. That is, $p^{*}\in\mathcal{M}\cap \Delta_{n}$. As a result, all trajectories starting from $p(0)\in\mathcal{M}$ exponentially converge to $p^{*}$.
\paragraph*{Proof of \ref{T6-1}} Let $p(s)\in\mathcal{H}$. Then for any $i\in\Vp$, by \eqref{e11} we have 
\begin{align*}
0<p_{i}(s+1)\leq \frac{a_{i}}{4}+\frac{1-a_{i}}{4}b_{i}+\frac{1-a_{i}}{n}\leq\frac{1}{2}, 
\end{align*}
where the last inequality is equivalent to \eqref{equ:15}, and the second inequality directly implies $0<p_{i}(s+1)\leq b_{i}/4+1/n$ for $i\in\Vf$. Hence, $\mathcal{H}$ is positively invariant under \eqref{equ:15}. Following the same contraction arguments as in \ref{T6-2}, we obtain  
\begin{align*}
  \norm{J(p(s))}_{1}=\max_{i\in\Vp}a_{i}(2 p_{i}(s)+ 1-2p_{i}(s) )=\max_{i\in\Vp}a_{i}<1,
\end{align*}
which completes the proof. 

\section{Proof of Theorems \ref{T3} and \ref{T4}}\label{app:T3}
\subsection{Proof of Theorem \ref{T3}}
Since $\mathcal{G}(C)$ is a star topology, we have $C_{j1}=1$ and $C_{ji}=0$ for all $i,j\neq 1$. By~\eqref{e11}, we obtain 
\begin{equation}\label{e17}
\begin{cases}
  p_{1}(s+1)=\sum_{j\in\Vp}\frac{a_{j}}{1-a_{j}}(1-p_{j}(s))p_{j}(s)+\frac{1}{n},\\
p_{i}(s+1)=a_{i}(p_{i}(s))^{2}+\frac{1-a_{i}}{n}, \quad \quad \forall i\in\Vp,
\end{cases}
\end{equation}
and $p_{i}(s+1)=1/n$ for all $i\in\Vf\setminus \{1\}$ and $s\geq 0$. 

\emph{Proof of \ref{T3-1}:} We first prove sufficiency. For $ i\in\Vp$ and any $|p_{i}(s)|<\xi_i$, we have 
\begin{align*}
  \xi_i>p_{i}(s+1)=a_{i}(p_{i}(s))^{2}+\frac{1-a_{i}}{n}>0.
\end{align*}
Hence, it suffices to focus on $p_{i}(0)\in (0,\xi_i)$. Define $\map{f_{i}}{(0,\xi_i)}{(0,\xi_i)}$ by 
\begin{align}\label{equ:31}
f_{i}(p_{i})=a_{i}p_{i}^{2}+\frac{1-a_{i}}{n},
\end{align}
then $p_{i}(s+1)=f_{i}(p_{i}(s))$ and $p_{i}=f_{i}(p_{i})$ in $(0,\xi_i)$ if and only if $p_{i}=\frac{1}{a_{i}}-\xi_i$.
Thus, $f_{i}$ has a unique fixed point 
\begin{equation}\label{equ:32}
p_{i}^{*}=\frac{1}{a_{i}}-\xi_i\in (\frac{1-a_{j}}{n},\frac{1}{n}).
\end{equation}
Moreover, by \eqref{equ:31} we have    
\begin{align*}
|f_{i}(p_{i})-p_{i}^{*}|=a_{i}(p_{i}+p_{i}^{*})|p_{i}-p_{i}^{*}|< |p_{i}-p_{i}^{*}|,
\end{align*}
where $a_{i}(p_{i}+p_{i}^{*})<1$ follows from \eqref{equ:32}. Consequently, for any $p_{i}(0)\in(-\infty,\xi_i)$, $p_{i}(s)$ converges to $p_{i}^{*}$. As a result, for any $p_{1}(0)\in\mathbb{R}$, $p_{1}(s)$ converges to 
\begin{align*}
p_{1}^{*}=\sum_{j\in\Vp}\frac{a_{j}}{1-a_{j}}(1-p_{j}^{*})p_{j}^{*}+\frac{1}{n}=\frac{1}{n}\sum_{j\in\Vp}^{n}\frac{a_{j}(1-p^{*}_{j})}{1-a_{j}p^{*}_{j}}+\frac{1}{n},
\end{align*}
where the last equality follows from 
{\small\begin{align}\label{equ:34}
p_{j}^{*}=\frac{1-a_{j}}{n(1-a_{j}p_{j}^{*})}, \quad \forall j\in\Vp.
\end{align}}
Finally, the necessity follows from $f_{i}(\xi_i)=\xi_i>1$ and $f_{i}(p_{i})> p_{i}$ for all $|p_{i}|> \xi_i$.

\emph{Proof of \ref{T3-2}:} Since $p_{1}(s+1)$ depends only on $p_{j}(s)$ with $j\in\Vp$, we assume, without loss of generality, that $\Vf=\{1\}$. 
Define $\map{F}{\mathbb{R}^{n}}{\mathbb{R}^{n}}$ by 
\begin{equation*}
  \begin{cases}
    F_{1}(p)=\sum_{j\in\Vp}\frac{a_{j}}{1-a_{j}}(1-p_{j})p_{j}+\frac{1}{n},\\
    F_{i}(p)=f_{i}(p), \quad \quad \forall i\in\Vp.
  \end{cases}
  \end{equation*}
Then, $F$ is differentiable on $\mathcal{Z}_{\beta}$, and its Jacobian is given by
\begin{equation*}
\begin{aligned}
  &\frac{\partial F}{\partial p}(p)=&\begin{bmatrix}
    0 & \frac{a_{2}}{1-a_{2}}(1-2p_{2}) & \dots & \frac{a_{n}}{1-a_{n}}(1-2p_{n})\\
    0 & 2a_{2}p_{2} & \dots & 0\\
    \vdots & \vdots & \vdots & \vdots\\
    0 & 0 & \dots & 2a_{n}p_{n}
 \end{bmatrix}.  
 \end{aligned}
\end{equation*}
Therefore, $\rho(\partial F/\partial p(p))=\max_{j}2a_{j}p_{j}$. Since 
\begin{equation*}
2a_{j}p^{*}_{j}=1-\sqrt{1-4a_{j}(1-a_{j})/n}<1, \; \forall j\neq 1,
\end{equation*}
we have $\rho(\partial F/\partial(p^{*}))<1$. Hence, $p^{*}$ is exponentially stable.

\subsection{Proof of Theorem~\ref{T4}}\label{app:T4}
Since the center node is partially stubborn, by~\eqref{e11} we have 
{\small\begin{equation}\label{e20}
 \begin{cases}
  \begin{aligned}
    p_{1}(s+1)=&(1-a_{1})\sum_{j\in\Vp\setminus\{1\}}\frac{a_{j}}{1-a_{j}}(1-p_{j}(s))p_{j}(s)\\
    &+a_{1}(p_{1}(s))^{2}+\frac{1-a_{1}}{n},
  \end{aligned}\\
p_{i}(s+1)=a_{i}(p_{i}(s))^{2}+\frac{1-a_{i}}{n}, \quad i\in\Vp\setminus\{1\}\\
p_{i}(s+1)=C_{1i}\frac{a_{1}}{1-a_{1}}(1-p_{1}(s))p_{1}(s)+\frac{1}{n}. \quad i\in\Vf
\end{cases}
\end{equation}
}


\emph{Proof of \ref{T4-2}:} Note that since $C$ is row-stochastic and zero-diagonal, this is equivalent to the case of fully stubborn center node when $n=2$. We assume that $n\geq 3$.

Let $\mathcal{Z}=\setdef{z\in\mathbb{R}^{n}}{|z_1|<1/(2a_1), z_{i}\in [0,1], \forall i\in\Vp\setminus\{1\}}$. For any $p(s)\in\mathcal{Z}$, we have $0<p_{1}(s+1)$, and 
\begin{align*}
  p_{1}(s+1)< \frac{1}{4a_{1}}+(1-a_{1})(\frac{1}{4}\sum_{j\in\Vp\setminus\{1\}}\frac{a_{j}}{1\!-\!a_{j}}+\frac{1}{n})\leq \frac{1}{2a_{1}},
\end{align*}
where the last inequality is implied by \eqref{equ:17}. For $i\in\Vp\setminus\{1\}$, we obtain 
\begin{align*}
  0<p_{i}(s+1)\leq a_{i}+\frac{1-a_{i}}{n}<1, \forall i\in\Vp\setminus\{1\}.
\end{align*}
Hence, $\mathcal{Z}$ is positively invariant. 

By \eqref{e20}, the subsystem consisting of all $i\in\Vp$ is closed; we therefore  
focus only on $i\in\Vp$. Suppose that $\Vp=\until{m}$ with $m<n$ and $z(s)=(p_{1}(s),\dots,p_{m}(s))^{\top}$. Let $\hat{\mathcal{Z}}=\varGamma_{m}(\mathbf{0}_{m},1/(2a_{1})\mathbf{e}_{1}+\sum_{j=2}^{m}\mathbf{e}_{j})$. Then, the invariance of $\mathcal{Z}$ implies that $z(s)\in\hat{\mathcal{Z}}$ for any $s\geq 1$ and $p(0)\in\mathcal{Z}$.
For $\eta\in (0,1)$, let $\hat{\mathcal{Z}}_{\eta}=\varGamma_{m}(\sum_{j\in\Vp}\frac{1-a_{j}}{n}\mathbf{e}_{j},\frac{1-\eta}{2a_{1}}\mathbf{e}_{1}+\frac{1}{n}\sum_{j=2}^{m}\mathbf{e}_{j})$.
By the proof of Theorem \ref{T3}\ref{T3-1}, for $j\in\Vp\setminus\{1\}$ and any $z_{j}(0)\in[0,1]$, $z_{j}(s)$ converges to $p_{j}^{*}$ given by \eqref{equ:32}. As a result, there exists $T(p(0))>0$ such that for all $s>T(p(0))$, $z_{j}(s-1)\in ((1-a_{j})/n,1/n)$. Let $\theta=1/4-(n-1)/n^{2}$. Since $n>2$, we have $\theta>0$, and $z_{j}(s-1)(1-z_{j}(s-1))<(n-1)/n^{2}=1/4-\theta$ for all $j\in\Vp\setminus\{1\}$. Thus, for $s>T(p(0))$, $p(0)\in\mathcal{Z}$ and $0<\eta\leq 2\theta(n-4a_{1}(1-a_{1}))/n<1/2$, 
\begin{multline*}
 \frac{1-a_{1}}{n}<z_{1}(s)<\frac{1}{4a_{1}}+\frac{1-a_{1}}{n}\\+(1-a_{1})\frac{1-4\theta}{4}\sum_{j\in\Vp\setminus\{1\}}\frac{a_{j}}{1\!-\!a_{j}}\leq\frac{1\!-\!\eta}{2a_{1}}, 
\end{multline*}
where the last inequality follows from \eqref{equ:17}. Therefore, for any $p(0)\in\mathcal{Z}$, there exists $T(p(0))>0$ and $\eta<1/2$ such that $z(s)\in\interior\hat{\mathcal{Z}}_{\eta}$ for all $s>T(p(0))$.
Define $\map{F}{\hat{\mathcal{Z}}_{\eta}}{\hat{\mathcal{Z}}_{\eta}}$ such that $z(s+1)=F(z(s))$. By \eqref{e20}, $F$ is differentiable on $\interior\hat{\mathcal{Z}}_{\eta}$ and continuous on $\hat{\mathcal{Z}}_{\eta}$. For an infinitesimal displacement $\delta z(s)$, we have $\delta z(s+1)=\frac{\partial F}{\partial z}(z(s))\delta z(s)$. Let $A_{m}=\diag(a_{1},\dots,a_{m})$ and $\delta \tilde{z}(s)=(I_{n}-A_{m})^{-1}\delta z(s)$. Then, $\delta \tilde{z}(s+1)=G(z(s))\delta \tilde{z}(s)$ with 
\begin{align}\label{equ:39}
  G(z)=&(I_{m}-A_{m})^{-1}\frac{\partial F}{\partial z}(z)(I_{m}-A_{m})\notag\\
  =&\begin{bmatrix}
    2a_{1}z_{1} & a_{2}(1-2z_{2}) & \dots & a_{m}(1-2z_{m})\\
    0 & 2a_{2}z_{2} & \dots & 0\\
    \vdots & \vdots & \vdots & \vdots\\
    0 & 0 & \dots & 2a_{m}z_{m}
 \end{bmatrix}.  
\end{align}

For $z\in\hat{\mathcal{Z}}_{\eta}$, since $1/2>z_{j}>0$ for all $j>1$ and $0<z_{1}\leq\frac{1-\eta}{2a_{1}}$ with $0<\eta<1/2$, we have $||G(p)||_{1}\leq \varepsilon$ with $\varepsilon=\max\{1-\eta,\max_{j\in\Vp\setminus\{1\}}a_{j}\}<1$. Thus, $\hat{\mathcal{Z}}_{\eta}$ is a generalized contraction region of $z(s+1)=F(z(s))$, and all trajectories of $z(s+1)=F(z(s))$ starting from $\hat{\mathcal{Z}}_{\eta}$ exponentially converge to a unique eqiulibrium $z^{*}$ satisfying $z_{j}^{*}=p_{j}^{*}$ for all $1<j\leq m$. As a result, for any $p(0)\in\mathcal{Z}$, $p(s)$ converges to a unique equilibrium $p^{*}$. 

It remains to show $p^{*}\in\interior\Delta_{n}$. By \eqref{e20} and \eqref{equ:34}, we have 
\begin{equation}\label{e21}
p_{1}^{*}=a_{1}(p_{1}^{*})^{2}+\frac{1-a_{1}}{n}+\frac{1-a_{1}}{n}\sum_{j\in\Vp\setminus\{1\}}(1-np_{j}^{*}),
\end{equation}
which, in $\mathcal{Z}$, is solved by
\begin{align}\label{equ:21.0}
p_{1}^{*}=\frac{1-\sqrt{1-\frac{4a_{1}(1-a_{1})}{n}(\mid\Vp\mid-n\sum_{j\in\Vp\setminus\{1\}}p_{j}^{*})}}{2a_{1}}.
\end{align}
On the other hand, \eqref{e21} implies 
\begin{multline*}
\frac{a_{1}(1-p_{1}^{*})}{1-a_{1}}p_{1}^{*}=\frac{1-a_{1}p_{1}^{*}}{1-a_{1}}p_{1}^{*}-p_{1}^{*}=\\\frac{1}{n}+\frac{1}{n}\sum_{j\in\Vp\setminus\{1\}}(1-np_{j}^{*})-p_{1}^{*}=\frac{\mid \Vp\mid}{n}-\sum_{j\in\Vp}p_{j}^{*}.
\end{multline*}
Hence, for $j\in\Vf$, we have 
\begin{align*}
p_{j}^{*}=\frac{1}{n}+C_{1j}\frac{a_{1}(1-p_{1}^{*})}{1-a_{1}}p_{1}^{*}
=\frac{1}{n}\!+\!C_{1j}(\frac{\mid\Vp\mid}{n}-\sum_{j\in\Vp}p_{j}^{*}).
\end{align*}
Consequently, $p^{*}\in\interior\Delta_{n}\intersection\mathcal{Z}$, which completes the proof.  

\emph{Proof of \ref{T4-3}:} Suppose that $\Vp=\until{m}$ and define $\map{F}{\mathcal{Z}}{\mathcal{Z}}$ such that $p(s+1)=F(p(s))$. 
Then, $F$ is differentiable on $\interior\mathcal{Z}$. Let $H(p)=(I_{n}-A)^{-1}\frac{\partial F}{\partial p}(p)(I_{n}-A)$, where $\frac{\partial F}{\partial p}(p)$ is the Jocobian of $F$. Then, we have $\rho(\frac{\partial F}{\partial p}(p))=\rho(H(p))$ and $H(p)=[\hat{H}(p) \ \  \mathbf{0}_{n\times (n-m)}]$ with
\begin{equation*}
  \hat{H}(p)=\begin{bmatrix}
    G(z)\\
    \begin{smallmatrix}
   C_{1m+1}a_{1}(1-2p_{1}) & 0 & \dots & 0\\
    \vdots & \vdots & \vdots & \vdots \\
    C_{1n}a_{1}(1-2p_{1}) & 0 & \dots & 0\\
    \end{smallmatrix}
 \end{bmatrix},
\end{equation*}
where $G(z)$ is given by \eqref{equ:39} with $z_i=p_i$ for $i\in\Vp$. Thus, 
\begin{multline*}
    \norm{H(p^{*})}_{1}=\norm{\hat{H}(p^{*})}_{1}=\max_{j\in\Vp}(2a_{j}p_{j}^{*}+a_{j}\mid1-2p^{*}_{j}\mid).
\end{multline*}
Since $p^{*}_{j}<1/2$ for all $m\geq j>1$, we obtain $\norm{H(p^{*})}_{1}=\amax<1$ when $p_{1}^{*}\leq 1/2$. 
Otherwise,~\eqref{equ:21.0} yields 
\begin{align*}
 &2a_{1}p_{1}^{*}+a_{1}\mid1-2p^{*}_{1}\mid=4a_{1}p_{1}^{*}-a_{1}\\=&2-a_{1}-2\sqrt{1-\frac{4a_{1}(1-a_{1})}{n}(m-n\sum_{j\in\Vp\setminus\{1\}}p_{j}^{*})}\\
 <&2-a_{1}-2\sqrt{1-\frac{4a_{1}(1-a_{1})}{n}(1+\sum_{j\in\Vp\setminus\{1\}}a_{j})}<1
\end{align*}
where the first inequality is implies by $p_{j}^{*}>(1-a_{j})/n$ for all $j\in\Vp\setminus\{1\}$, and the last inequality follows from \eqref{equ:17} and Lemma~\ref{L4}. In conclusion, we have $\rho(\frac{\partial F}{\partial p}(p^{*}))<1$, which means that $p^{*}$ is exponentially stable. 

\section{Proof of Theorem \ref{T5}}\label{app:T5} 

\emph{Proof of \ref{T5-1}:} Let $\chi (s)=\mathbf{1}^{\top}_{n}p(s)$, by \eqref{ehomo} we have 
\begin{align*}
&\chi  (s+1)\\
=&a\sum_{i=1}^{n}(p_{i}(s))^{2}+a\sum_{j=1}^{n}p_{j}(s)(1-p_{j}(s))\sum_{i=1}^{n}C_{ji}+(1-a)\\
=& a \chi (s)+1-a.
\end{align*}
As a result, $\chi (s)\to 1$ as $s\to\infty$ for any $\chi(0)\in\mathbb{R}$. 
For any $s\geq 0$ and $p(s)\in\Delta_{n}$, we have $\chi (s+1)=a \chi (s)+1-a=1$. Moreover, \eqref{ehomo} implies $p_{i}(s+1)\geq (1-a)/n$ and 
\begin{multline*}
p_{i}(s+1)< ap_{i}(s)+a\sum_{j\neq i}p_{j}(s)+\frac{1-a}{n}=a+\frac{1-a}{n}<1.
\end{multline*}
Thus, $\Delta_{n}$ is positively invariant.

\emph{Proof of \ref{T5-2}:} By the proof of \ref{T5-1}, $0<p_{i}(s)<a+(1-a)/n$ for all $i$ and $s>0$. Therefore, for any $s>0$, by \eqref{ehomo} we have 
\begin{align*}
 &\norm{p(s+1)-p^{*}}_{1}\\=&a\sum_{i=1}^{n}\mid(p_{i}(s)+p_{i}^{*})(p_{i}(s)-p_{i}^{*})\\&+\sum_{j=1}^{n}C_{ji}(p_{j}(s)-p_{j}^{*})(1-p_{j}(s)-p_{j}^{*})\mid \\
 \leq & a\sum_{i=1}^{n}(p_{i}(s)+p_{i}^{*}+\mid 1-p_{i}(s)-p_{i}^{*}\mid)\mid p_{i}(s)-p_{i}^{*}\mid\\
 \leq & a \max_{i}(p_{i}(s)+p_{i}^{*}+\mid 1-p_{i}(s)-p_{i}^{*}\mid)\norm{p(s)-p^{*}}_{1}\\
 <& a (4(a+\frac{1-a}{n})-1)\norm{p(s)-p^{*}}_{1}< \norm{p(s)-p^{*}}_{1},
\end{align*}
where the last inequality follows from \eqref{equ:19}. Hence, $p(s)$ exponentially converges to $p^{*}$.

\bibliographystyle{plainurl}
\bibliography{alias,Ye}

@article{CB-LW-MF-CA:23,
  title =	 {Quantifying leadership in climate negotiations: {A} social power game},
  author =	 {C. Bernardo and L. Wang and M. Fridahl and C. Altafini},
  journal = {PNAS Nexus},
    volume = {2},
    number = {11},
    pages = {pgad365},
    year = {2023},
    doi = {10.1093/pnasnexus/pgad365}
}

@article{DFG:15,
author = {D. F. Gleich},
title = {Page{R}ank Beyond the Web},
journal = {SIAM Review},
volume = {57},
number = {3},
pages = {321-363},
year = {2015},
doi = {10.1137/140976649},
}

@article{YT-LW:23,
  author =	 {Y. Tian and L. Wang},
  title =	 {Dynamics of Opinion Formation, Social Power Evolution and Naïve Learning in Social Networks},
  journal ={Annual Reviews in Control}, 
  volume =	 55,
  pages =	 {182-193},
  year =	 2023,
  doi =		 {10.1016/j.arcontrol.2023.04.001}
  }

@article{RB:50 ,
  author =	 {R. Bierstedt},
  title =	 {An analysis of social power},
  journal ={American Sociological Review}, 
  volume =	 15,
  pages =	 {730-738},
  year =	 1950,
  doi =		 {10.2307/2086605}
  }

@article{MG-WBB-JD-SLF-FK-HO-DP-DLS-TD:21 ,
  author =	 {M. Galesic and W. B. de~Bruin and J. Dalege and S. L. Feld and F. Kreuter and
   H. Olsson and D. Prelec and D. L. Stein and T. van~de~Does},
  title =	 {Human social sensing is an untapped resource for computational social science},
  journal =nature, 
  volume =	 595,
  pages =	 {214-222},
  year =	 2021,
  doi =		 {10.1038/s41586-021-03649-2}
  }

@InCollection{AZ-ARC-ES:59,
  title =	 {Power and the relations among professions},
  author =	 {A. Zander and A. R. Cohen and E. Stotland},
  booktitle =	 {Studies in Social Power},
  pages =	 {15--34},
  year =	 {1959}
}

@article{YT-LW-FB:23,
  author =	 {Y. Tian and L. Wang and F. Bullo},
  journal =	 sicon,
  title =	 {How social influence affects the wisdom of crowds in influence networks},
  volume =	61,
  number =	 4,
  pages =	 {2334-2357},
  year =	 2023,
  doi =		 {10.1137/22M1492751}
}

@article{HI-RT:14,
  title =	 {The {PageRank} Problem, Multiagent Consensus, and Web
                  Aggregation: {A} Systems and Control Viewpoint},
  author =	 {H. Ishii and R. Tempo},
  journal =	 csmOLD,
  volume =	 34,
  number =	 3,
  pages =	 {34--53},
  year =	 2014,
  doi =		 {10.1109/MCS.2014.2308672},
}

@article{XC-JL-MAB-ZX-TB:16,
  author =	 {X. Chen and J. Liu and M.-A. Belabbas and Z. Xu and T. Ba{\c{s}}ar},
  journal =	 tac,
  title =	 {Distributed Evaluation and Convergence of Self-Appraisals in Social Networks},
  volume =	62,
  number =	 1,
  pages =	 {291-304},
  year =	 2017,
  doi =		 {10.1109/TAC.2016.2554280},
}

@ARTICLE{10506633,
  author={L. Wang and G. Chen and C. Bernardo and Y. Hong and G. Shi and C. Altafini},
  journal={IEEE Transactions on Automatic Control}, 
  title={Social Power Games in Concatenated Opinion Dynamics}, 
  year={2024},
  volume={69},
  number={11},
  pages={7614-7629},
  doi={10.1109/TAC.2024.3391870}}

@article{JK-JO-AS-PBG:25,
author = {J. Kov{\'a}\v{r}{\'i}k  and J. Ozaita  and A. S{\'a}nchez  and P. Bra{\~n}as-Garza },
title = {Perception of own centrality in social networks},
journal = {Proceedings of the National Academy of Sciences},
volume = {122},
number = {47},
pages = {e2420334122},
year = {2025},
doi = {10.1073/pnas.2420334122}
}

@Book{GHH-JEL-GP:34,
  author =	 {G. H. Hardy and J. E. Littlewood and G. P\'{o}lya},
  title={Inequalities},
  year={1934},
  publisher={Cambridge University Press},
  address={Cambridge},
}

@article{JRPF:56,
  Author =	 {J. R. P. {French~Jr.}},
  Title =	 {A formal theory of social power},
  Journal =	 {Psychological Review},
  Pages =	 {181--194},
  Volume =	 63,
  Number =	 3,
  year =	 1956,
  fbnote =	 "was JF:56",
 doi =		 {10.1037/h0046123},
}

@article{FB-FF-BF:20,
  title =	 {Finite-time influence systems and the Wisdom of Crowd effect},
  author =	 {F. Bullo and F. Fagnani and B. Franci},
  journal =	 sicon,
  volume =	 58,
  pages =	 {636-659},
  year =	2020,
  doi =		 {10.1137/18M1232267}
}

@article{YT-PJ-AM-LW-NEF-FB:22,
  author =	 {Y. Tian and P. Jia and A. MirTabatabaei and L. Wang and
                  N. E. Friedkin and F. Bullo},
  title =	 {Social power evolution in influence networks with
                  stubborn individuals},
  journal =	 tac,
  volume =	 67,
  number=2,
  pages =	 {574-588},
  year =	 2022,
   doi =		 {10.1109/TAC.2021.3052485}
}

@Article{NEF-PJ-FB:16,
  author =	 {N. E. Friedkin and P. Jia and F. Bullo},
  title =	 "A theory of the evolution of social power: {N}atural
                  trajectories of interpersonal influence systems along
                  issue sequences",
  journal =	 "Sociological Science",
  year =	 2016,
  volume =	 3,
  number=20,
  pages =	 {444-472},
 doi =		 {10.15195/v3.a20},

}

@Article{MY-JL-BDOA-CY-TB:17b,
  author =	 {M. Ye and J. Liu and B. D. O. Anderson and C. Yu and
                  T. Ba\c{s}ar},
  title =	 {Evolution of Social Power in Social Networks with Dynamic
                  Topology},
  journal =	 tac,
  volume =	 63,
  number =	 11,
  pages =	 {3793-3808},
  year =	 2018,
  doi =		 {10.1109/TAC.2018.2805261},
  }

@Book{DC:59,
  title =	 {Studies in Social Power},
  author =	 {D. Cartwright},
  year =	 {1959},
  publisher ={University of Michigan},
  address={Ann Arbor},
}

@article{NEF:11,
  title =	 {A formal theory of reflected appraisals in the evolution
                  of power},
  author =	 {N. E. Friedkin},
  journal =	 {Administrative Science Quarterly},
  volume =	 56,
  number =	 4,
  pages =	 {501-529},
  year =	 2011,
 doi =		 {10.1177/0001839212441349},
}

@Article{PJ-AM-NEF-FB:13d,
  author =	 {P. Jia and A. MirTabatabaei and N. E. Friedkin and
                  F. Bullo},
  title =	 {Opinion dynamics and the evolution of social power in
                  influence networks},
  journal =	 sirev,
  year =	 2015,
  volume =	 57,
  number =	 3,
  pages =	 {367-397},
  doi =		 {10.1137/130913250},
}

@article{PJ-NEF-FB:14e,
  author =	 {P. Jia and N. E. Friedkin and F. Bullo},
  title =	 {Opinion dynamics and social power evolution: {A} single-timescale model},
  journal =	 tcns,
  year =	 2019,
  volume =	 7,
  number=2,
  pages =	 {899-911},
  doi =		 {10.1109/TCNS.2019.2951672},
}

@article{AVP-RT:18,
  title =	 {A Tutorial on Modeling and Analysis of Dynamic Social
                  Networks. {Part II}},
  author =	 {A. V. Proskurnikov and R. Tempo},
  journal =	 {Annual Reviews in Control},
  volume =	 45,
  pages =	 {166-190},
  doi =		 {10.1016/j.arcontrol.2018.03.005},
  year =	 2018
}

@Article{BG-MOJ:10,
  author =	 {B. Golub and M. O. Jackson},
  title =	 {Na\"ive Learning in Social Networks and the Wisdom of
                  Crowds},
  journal =	 {American Economic Journal: Microeconomics},
  year =	 2010,
  volume =	 2,
  number =	 1,
  pages =	 {112-149},
  doi =		 {10.1257/mic.2.1.112},
}

@Article{MHDG:74,
  author =	 {M. H. DeGroot},
  title =	 {Reaching a Consensus},
  journal =	 {Journal of the American Statistical Association},
  year =	 1974,
  volume =	 69,
  number =	 345,
  pages =	 {118-121},
  doi =		 {10.1080/01621459.1974.10480137}
}

@article{YT-LW:18,
  title =	 {Opinion dynamics in social networks with stubborn agents:
                  {A}n issue-based perspective},
  author =	 {Y. Tian and L. Wang},
  journal =	 automatica,
  volume =	 96,
  pages =	 {213-223},
  year =	 2018,
  doi =		 {10.1016/j.automatica.2018.06.041}
}

@InCollection{NEF-ECJ:99,
  title =	 {Social influence networks and opinion change},
  author =	 {N. E. Friedkin and E. C. Johnsen},
  booktitle =	 {Advances in Group Processes},
  volume =	 16,
  pages =	 {1-29},
  year =	 1999,
  publisher =	 {Emerald Group Publishing Limited},
  editor =	 {S. R. Thye and E. J. Lawler and M. W. Macy and
                  H. A. Walker},
  ISBN =	 0762304529
}

@Book{RAH-CRJ:12,
  author =	 {R. A. Horn and C. R. Johnson},
  title =	 {Matrix Analysis},
  publisher =	 cambridge,
  year =	 2012,
  isbn =	 0521548233,
  edition =	 {2nd},
}

@STRING{tac  = "IEEE Transactions on Automatic Control"}

@STRING{tcns  = "IEEE Transactions on Control of Network Systems"}

@String{ar = "Autonomous Robots"}

@STRING{sirev = "SIAM Review"}

@STRING{sicon = "SIAM Journal on Control and Optimization"}

@STRING{automatica = "Automatica"}

@string{csmOLD = "{IEEE} Control Systems Magazine"}

@string{nature = "Nature"}

@string{science = "Science"}

@String{SIAM = "SIAM"}

@STRING{cambridge = "Cambridge University Press"}

@String{pnas = "Proceedings of the National Academy of Sciences"}

\end{document}